\documentclass[lettersize,journal]{IEEEtran}
\usepackage{array}
\usepackage[caption=false, font=normalsize, labelfont=sf, textfont=sf]{subfig}
\usepackage{textcomp}
\usepackage{url}
\usepackage{amsthm}
\usepackage{tikz}
\usepackage{verbatim}
\usepackage{diagbox}
\usepackage{cite}
\usepackage{amsmath,amssymb,amsfonts}
\usepackage{algorithmic}
\usepackage{graphicx}
\usepackage[ruled,linesnumbered]{algorithm2e}
\usepackage{textcomp}
\usepackage{multirow}
\usepackage{xcolor}
\usepackage{stfloats}

\newtheorem{remark}{Remark}
\newtheorem{definition}{Definition}
\newtheorem{proposition}{Proposition}

\newtheorem{example}{Example}
\newcounter{problem}
\renewcommand{\theproblem}{\textbf{P\arabic{problem}}}
\allowdisplaybreaks

\begin{document}

\title{Positional Attention-based Graph Neural Network for Learning Permutation Non-equivariant Wireless Policies}

\author{Baichuan Zhao, Chenyang Yang, Jianyu Zhao and Di Zhang
\thanks{The authors are with the School of Electronics and Information Engineering, Beihang University, Beijing 100191, China (E-mail: \{zhaobaichuan; cyyang; di\_zhang\}@buaa.edu.cn).}
}

\maketitle

\textcolor{blue}{ABSTRACT, INTRODUCTION AND CONCLUSION REMAIN UNREVISED.}

\begin{abstract}
Graph neural networks (GNNs) have emerged as a promising approach to learning wireless policies efficiently by leveraging topology prior and incorporating relational inductive biases.
\textcolor{black}{However, when the optimal policy is not permutation equivariant (PE), conventional GNNs suffer from mismatched inductive biases, leading to degraded performance or poor generalizability.}
This issue arises in wireless tasks with expected objectives, such as channel estimation and end-to-end (E2E) precoding, where the PE property of the optimal policy depends on the underlying channel distribution.
In this paper, we propose a novel positional attention-based GNN to learn permutation non-equivariant policies efficiently.
The core idea is to incorporate relative positions of vertices into the attention mechanism via an embedding function, enabling the GNNs to capture asymmetric relationships.
\textcolor{black}{Consequently, the proposed GNN can represent permutation non-equivariant functions, while retaining high learning efficiency and size generalizability through parameter sharing.}
We consider channel estimation and E2E precoding as case studies, and prove that their policies are PE to users but not to antennas under spatially correlated channels.
We employ the proposed GNN to learn the policies, where the embedding function is designed based on the channel covariance matrix.
Simulation results demonstrate that the proposed GNN outperforms existing channel estimation and E2E precoding methods, requires fewer samples for training, and can be generalized to systems with different numbers of antennas and users.
\end{abstract}


\section{Introduction}

Graph neural networks (GNNs) have recently emerged as a powerful approach for learning policies efficiently in wireless communication systems.
\textcolor{black}{By introducing parameter sharing schemes, GNNs can harness the permutation equivariant (PE) property of wireless policies, thereby reducing the hypothesis space (i.e., the set of functions representable by the GNN) and training complexity \cite{zhao2024understanding}.}
Moreover, the parameter sharing scheme allows a well-trained GNN to infer on different problem sizes without retraining \cite{lu2024graph}.
To improve the expressive power of GNNs, i.e., to enlarge their hypothesis space, attention-based GNNs have been proposed \cite{velickovic2017graph, dwivedi2020generalization}, which have been applied to a wide range of wireless communication tasks \cite{huang2024sub, yang2024graph, liu2024multidimensional, guo2024recursive, mishra2024graph, hu2024unsupervised, wang2024graph} and achieved better performance.

Despite the wide application of GNNs, existing GNN architectures face an inherent trade-off between learning efficiency and performance when the optimal policy is permutation non-equivariant, i.e., its output is not equivariant to input permutations along some dimensions.
On one hand, \textcolor{black}{conventional GNNs adopt parameter sharing to improve learning efficiency, which, however, inherently enforces permutation equivariance, leading to a biased hypothesis space and degraded performance.}
On the other hand, removing parameter sharing enlarges the hypothesis space but incurs higher training complexity and hinders generalization across varying system sizes.

In wireless communication systems, this trade-off arises when the optimization objective of a wireless problem involves mathematical expectations, as to be shown later.
For instance, in channel estimation, the widely used mean square error (MSE) metric contains an expectation over the true channel.
Similarly, in precoding problems without accurate CSI, such as robust precoding \cite{shi2021deep} and end-to-end (E2E) precoding with implicit channel acquisition \cite{jiang2021learning}, precoders are designed by optimizing average performance metrics.
In these cases, the optimal solution depends on the underlying distribution, but whether the resulting policy exhibits permutation property under specific distributions remains unexplored.
As a result, when the policies are not PE, conventional GNN architectures may struggle to achieve satisfactory performance with low training complexity.

\subsection{Related Works}

\subsubsection{Attention-based GNNs for Wireless Communications}

Attention mechanisms allow GNNs to adaptively adjust the importance of neighboring vertices, enabling them to focus on most relevant neighbors and better capture dependencies during aggregation \cite{velickovic2017graph}.
Motivated by these advantages, attention-based GNNs has been used for various applications in wireless communications, such as precoding \cite{huang2024sub, yang2024graph, liu2024multidimensional, guo2024recursive}, power control \cite{mishra2024graph} and power allocation \cite{hu2024unsupervised, wang2024graph}.

To better understand the limitations of existing attention mechanisms, we review how attention coefficients are typically computed and what information they depend on.
In terms of computation, existing mechanisms adopt either the dot-product attention introduced in the Transformer architecture \cite{vaswani2017attention} or an FNN-based approach as in \cite{bahdanau2014neural}.
The dot-product attention mechanism was incorporated into GNNs to measure the similarity between vertex features during aggregation in \cite{mishra2024graph, guo2023model}.
To enhance expressive power, the attention coefficient was extended from a scalar to a vector in \cite{liu2024multidimensional}, and the dot-product was replaced with a Hadamard product.
For FNN-based attention, one well-known example is the Graph Attention Network (GAT) \cite{velickovic2017graph}, which integrates the attention mechanism in \cite{bahdanau2014neural} into a homogeneous GNN and has been used for precoding \cite{yang2024graph}.
More recent works has extended the GAT to heterogeneous graphs for various wireless tasks, such as precoding \cite{guo2024recursive} and power allocation \cite{wang2024graph}.

From the perspective of input, most existing attention mechanisms compute attention coefficients based solely on the hidden representations from the previous layer \cite{mishra2024graph, liu2024multidimensional, velickovic2017graph, yang2024graph, guo2024recursive, wang2024graph}.
The GNN proposed in \cite{guo2023model} is an exception, where both hidden representations and input features were used.
In all cases, the inputs used to compute attentions remain equivariant under vertex permutations, and thus these attention mechanisms preserve the permutation properties of the underlying GNN architecture, as will be discussed later.
Therefore, such mechanisms do not enable the model to represent permutation non-equivariant functions, motivating the design of new attention mechanisms.

\subsubsection{Channel Estimation and E2E Precoding}

In the following, we briefly review existing methods used for channel estimation and E2E precoding, which serve as two case studies considered in this work.

\underline{\it Channel Estimation:}
To improve estimation accuracy by exploiting the sparsity of wireless channels, iterative algorithms such as orthogonal matching pursuit (OMP) \cite{lee2016channel}, approximate message passing (AMP) \cite{donoho2009message}, and sparse Bayesian learning (SBL) \cite{wipf2004sparse} were proposed, at the cost of increased computational complexity.
The algorithms were then used in \cite{liu2024sparse, ruan2023simplified, yang2024regularized}.
To further enhance accuracy and reduce the computational burden, deep neural networks (DNNs) have been introduced for channel estimation.
A line of work designed model-driven DNNs by unfolding iterative algorithms such as AMP and SBL \cite{ruan2023simplified, lv2025sparse}.
However, their performance is limited by the original algorithms and the computational complexity remains high.
Alternatively, data-driven DNNs do not rely on mathematical models, making them more flexible for various tasks and faster at inference.
In \cite{jiang2021dual} and \cite{gao2021attention}, convolutional neural networks (CNNs) were used to capture channel correlations, where angular-domain information was leveraged to enhance estimation accuracy.
To improve scalability and generalizability, GNNs have recently been introduced.
In \cite{ye2024gnn}, a GNN was used to estimate channels in reconfigurable intelligent surface (RIS)-aided systems from received sounding reference signals (SRSs).
In \cite{singh2023channel}, the full channel matrix was recovered from received SRSs of a subset of RIS element groups.
To capture the spatial correlations among different RIS elements, each group was modeled as a vertex in a graph, and a GNN cascaded with a transformer was employed.
Channel recovery was considered in \cite{liu2024pd}, where pilots transmitted over partial time and frequency resources were treated as different tokens and fed into a Transformer.

\underline{\it E2E Precoding:} To improve system performance, several works have proposed to employ DNNs for optimizing precoding in an E2E manner.
In frequency division duplex systems, this involves joint design of pilots, quantization feedback and precoding \cite{sohrabi2021deep, guo2024deep}.
In time division duplex (TDD) systems, downlink precoders are learned directly from uplink SRSs by leveraging channel reciprocity \cite{chen2024joint, park2024end, jiang2022accurate}, which is also the focus or our work.
Specifically, FNNs were employed in \cite{chen2024joint} and \cite{park2024end}, while a Transformer that treats SRSs from different subframes as different tokens was adopted in \cite{jiang2022accurate} to capture temporal correlation.
Although with improved performance, these DNNs can not be generalized to varying numbers of users, since their numbers of trainable parameters depend on that of users.
Therefore, they can not be applied for practical wireless systems where the user number changes dynamically.
To tackle the challenge, several works have employed GNNs for E2E precoding.
In \cite{jiang2021learning}, both the precoder at base station (BS) and the phase shifts of RIS were jointly optimized, while in \cite{wang2024learning} and \cite{linfu2025graph}, hybrid analog and digital precoders were jointly optimized for mmWave systems.
The GNNs designed in these works are PE to users, and can infer in scenarios with different numbers of users.

\underline{\it Summary:} Aforementioned works using GNNs for channel estimation, i.e., \cite{liu2024pd, ye2024gnn}, or E2E precoding, i.e., \cite{jiang2021learning, wang2024learning, linfu2025graph}, did not examine whether the permutation properties of the GNNs are consistent with the policies.
Specifically, in \cite{liu2024pd}, the graph Transformer was followed by an FNN, and its learned policy is without permutation property.
The GNNs in \cite{ye2024gnn, jiang2021learning, wang2024learning, linfu2025graph} are PE to users, enabling inference across varying numbers of users.
However, these works did not analyze whether the policies are PE to antennas, nor did they introduce parameter sharing along the antenna dimension.
This may be because performance degradation was empirically observed after sharing parameters.
Consequently, the number of trainable parameters grows with that of antennas, leading to high sample complexity for large antenna arrays and poor generalizability across antenna dimensions.

\subsection{Motivation and Contribution}

When learning permutation non-equivariant policies, existing GNNs face a trade-off between learning efficiency and performance.
For wireless problems involving expected objectives, such as channel estimation and E2E precoding, the permutation properties of the corresponding policies depend on the underlying channel distributions, whereas the dependency has not been investigated yet.
Under certain distributions, these policies are not PE and can not be efficiently learned by existing GNNs.


In this paper, we address this issue by proposing a novel positional attention-based GNN to learn permutation non-equivariant policies efficiently.
We consider learning wireless policies in multi-user multi-input single-output (MU-MISO) systems, while the approach can be extended to other communication systems.

Our main contributions are summarized as follows.

\begin{itemize}
\item We formulate a class of wireless problems whose objective involves an expectation, and analyze how the underlying probability distribution affects the PE property of the resulting policy.
We identify the conditions under which the policy exhibits permutation property.

\item We propose a positional attention mechanism for GNNs to learn permutation non-equivariant policies.
Unlike existing attention mechanisms, our method computes attention coefficients via an embedding function, which encodes the relative positions of vertices based on their indices.
This design enables the GNN to better capture correlations induced by, e.g., the physical locations of entities in communication systems, and allows it to represent permutation non-equivariant functions.

\item Using channel estimation and E2E precoding as case studies, we prove that the corresponding policies are PE to users but not to antennas in spatially correlated channels.
Inspired by the iterative approximation of linear minimum mean square error (LMMSE) estimator, we design the embedding function by resorting to the structure of channel covariance, and the resulting weight matrix of the GNN exhibits the same structure as the channel covariance matrix.

\item Extensive simulations show that the proposed GNN can achieve better performance than state-of-the-art counterparts for channel estimation and E2E precoding.
Moreover, it requires much fewer samples to achieve the same performance than other methods, and can be generalized to different numbers of antennas and users.
\end{itemize}

\textit{Notations:} $a_{i,j}$ or $[\mathbf{A}]_{i,j}$ is the element in the $i$-th row and the $j$-th column of $\mathbf{A}$.
$\mathbf{I}$ denotes the identity matrix. $(\cdot)^\mathsf{T}$ and $(\cdot)^\mathsf{H}$ denote transpose and conjugate transpose of matrices, respectively.
$||\cdot||$ denotes Frobenius norm of matrices. $\odot$ and $\otimes$ denote Hadamard product and Kronecker product of matrices, respectively.
$\mathbb{U}(a,b)$ denotes uniform distribution with range between $a$ and $b$.
$\mathcal{CN}(\mu,\sigma)$ denotes complex Gaussian distribution with mean $\mu$ and deviation $\sigma$.
$P_{\mathtt{x}}(x)$ denotes the probability density function (PDF) of variable $\mathtt{x}$, where the typewriter font $\mathtt{x}$ distinguishes the variable from its specific value $x$.
$\mathbf{0}$ and $\mathbf{1}$ are matrices whose elements are all zeros and ones, respectively.
$\mathfrak{Re}(\cdot)$ and $\mathfrak{Im}(\cdot)$ denote the real and imaginary parts of a complex number, respectively.
$\mathbf{a} || \mathbf{b}$ denotes the concatenation of vectors $\mathbf{a}$ and $\mathbf{b}$.

\section{Problem Formulation}

Consider a MU-MISO system, where a BS equipped with $N$ antennas serves $K$ single-antenna users. 
Since there often exists randomness (say receiver noise) during  decision-making, we consider a stochastic optimization problem involving mathematical expectation \cite{kall1994stochastic}.

Denote $\mathbf{X} \in \mathbb{C}^{N \times K \times D_x}$ as the parameters measured for making a decision, and $\mathbf{Y} \in \mathbb{C}^{N \times K \times D_y}$ as the optimization variable, where $D_x$ and $D_y$ represent the dimensions of vectors corresponding to each antenna-user pair. 
The first and second dimensions of $\mathbf{X}$ and $\mathbf{Y}$ correspond to antenna and user dimensions, respectively.
The problem is formulated as
\refstepcounter{problem}\label{prob_stochastic}
\begin{IEEEeqnarray}{rcl}\IEEEyesnumber\label{eq_general_optimization}
\theproblem: \quad \max_{\mathbf{Y}} & \quad & U(\mathbf{X}, \mathbf{Y}) = \mathbb{E}_{\mathtt{Z} | \mathtt{X}}[U^\prime(\mathbf{Y}, \mathbf{Z}) | \mathbf{X}] \IEEEyessubnumber\label{eq_objective_general}\\
s.t. & \quad & \mathbf{c}(\mathbf{X}, \mathbf{Y}) {\leq} \mathbf{0},\IEEEyessubnumber
\end{IEEEeqnarray}
where $U^\prime(\cdot, \cdot)$ is a utility function, $U(\cdot, \cdot)$ is an \textit{average utility} over a random variable $\mathbf{Z} \in \mathbb{C}^{N \times K \times D_z}$ (e.g., the unknown channel state information) conditional on $\mathbf{X}$, the conditional PDF is $P_{\mathtt{Z} | \mathtt{X}} (\mathbf{Z} | \mathbf{X})$, and $\mathbf{c}(\cdot, \cdot)$ is a constraint function.

Denote the mapping from $\mathbf{X}$ to the optimal solution $\mathbf{Y}$ of problem \ref{prob_stochastic} as $\mathbf{Y} = f(\mathbf{X})$, which is referred to as a \textit{policy}.
In the following, we aim to design a DNN to solve problem \ref{prob_stochastic} by learning the policy.

\section{Distribution-dependent and Augmented Policies}\label{section_dd_policy}

In this section, we analyze how the permutation property of the policy $\mathbf{Y} = f(\mathbf{X})$ is determined by the underlying data distribution. 
We show that, although the optimization problem is defined over unordered sets, the policy does not necessarily inherit permutation equivariance when the corresponding PDF itself is not permutation invariant (PI). 
Then, we construct an augmented policy by introducing additional features that encode positional information, and prove that the resulting policy is always PE regardless of the underlying distribution.

\subsection{Definition of Permutation Invariance and Equivariance}

To facilitate further analysis, we define the permutation properties of general functions with inputs $\mathbf{T}_1, \cdots, \mathbf{T}_S$, where each input $\mathbf{T}_s \in \mathbb{C}^{N \times K \times D_s}$ is a three-dimensional tensor whose first and second dimensions correspond to the antenna and user dimensions, respectively.
Let $\boldsymbol{\Pi}_\mathrm{A} \in \mathbb{R}^{N \times N}$ and $\boldsymbol{\Pi}_\mathrm{U} \in \mathbb{R}^{K \times K}$ denote the permutation matrices corresponding to antennas and users, respectively, and the PI and PE properties are defined as follows.

\begin{definition}[\textbf{Permutation Invariance}]
A function $F(\mathbf{T}_1, \ldots, \mathbf{T}_S)$ is PI to antennas if
$
F(\mathbf{T}_1, \ldots, \mathbf{T}_S)
=
F(
\mathbf{T}_1\boldsymbol{\Pi}_\mathrm{U},
\ldots,
\mathbf{T}_S\boldsymbol{\Pi}_\mathrm{U}
)$ and PI to users if
$F(\mathbf{T}_1, \ldots, \mathbf{T}_S)
=
F(
\boldsymbol{\Pi}_\mathrm{A}^\mathsf{T}\mathbf{T}_1,
\ldots,
\boldsymbol{\Pi}_\mathrm{A}^\mathsf{T}\mathbf{T}_S
)$.\footnote{The matrix–tensor multiplication applies standard matrix multiplication to each slice along the third dimension.}
\end{definition}

\begin{definition}[\textbf{Permutation Equivariance}]
A tensor-valued function $\mathbf{F}(\mathbf{T}_1, \ldots, \mathbf{T}_S)
\in \mathbb{C}^{N \times K \times D_\mathrm{o}}$ is PE to antennas if
$
\boldsymbol{\Pi}_\mathrm{A}^\mathsf{T}
\mathbf{F}(\mathbf{T}_1, \ldots, \mathbf{T}_S)
=
\mathbf{F}(
\boldsymbol{\Pi}_\mathrm{A}^\mathsf{T}\mathbf{T}_1,
\ldots,
\boldsymbol{\Pi}_\mathrm{A}^\mathsf{T}\mathbf{T}_S
)
$
and PE to users if
$
\mathbf{F}(\mathbf{T}_1, \ldots, \mathbf{T}_S)
\boldsymbol{\Pi}_\mathrm{U}
=
\mathbf{F}(
\mathbf{T}_1\boldsymbol{\Pi}_\mathrm{U},
\ldots,
\mathbf{T}_S\boldsymbol{\Pi}_\mathrm{U}
).
$
\end{definition}


\subsection{Distribution-dependent Policy}

In this subsection, we will investigate whether the policy $\mathbf{Y} = f(\mathbf{X})$ inherits PE property from the unordered nature of the antenna and user sets. 
While such equivariance has been established in many deterministic wireless optimization problems, it does not hold for the stochastic problem \ref{prob_stochastic}. 
Specifically, we show that the permutation property of the policy depends on the permutation invariance of the underlying conditional distribution.

Problem \ref{prob_stochastic} involves two sets: the antenna set with $N$ BS antennas and the user set with $K$ users. 
Since these sets are inherently unordered, permuting their indices does not change the physical environment. 
For the policy $\mathbf{Y}=f(\mathbf{X})$, such permutations correspond to permuting the antenna and user dimensions of $\mathbf{X}$ and $\mathbf{Y}$, which suggests that the policy may be PE to antennas and users, i.e.,
\begin{equation}
\boldsymbol{\Pi}_\mathrm{A}^\mathsf{T} \mathbf{Y} \boldsymbol{\Pi}_\mathrm{U}
=
f\!\left(
\boldsymbol{\Pi}_\mathrm{A}^\mathsf{T}
\mathbf{X}
\boldsymbol{\Pi}_\mathrm{U}
\right).
\label{eq_policy_pe_expected}
\end{equation}

This intuition is valid for many deterministic wireless optimization problems, which can be regarded as a special case of \eqref{eq_general_optimization} where no latent random variable is involved. 
In such cases, the optimization problem does not involve any underlying statistical distribution, and the objective and constraints are fully determined by the observed input $\mathbf{X}$ and the decision variable $\mathbf{Y}$.

For example, consider the classical sum-rate maximization precoding problem under a transmit power constraint, i.e.,
\refstepcounter{problem}\label{prob_classic_precoding}
\begin{IEEEeqnarray}{rcl}\IEEEyesnumber\label{eq_e2e_optimization_re}
\theproblem: \quad \max_{\mathbf{V}} & \quad & R\left(\mathbf{H}_\mathrm{DL}, \mathbf{V}\right) \IEEEyessubnumber\\
s.t. & \quad & ||\mathbf{V}||_\mathrm{F}^2 - P_\mathrm{DL} \leq 0,\IEEEyessubnumber
\end{IEEEeqnarray}
where $\mathbf{H}_\mathrm{DL}, \mathbf{V} \in \mathbb{C}^{N \times K}$ denote the downlink channel matrix and precoding matrix, respectively, $R(\mathbf{H}_\mathrm{DL}, \mathbf{V})$ represents the achievable downlink sum rate, and $P_\mathrm{DL}$ is the total transmit power budget.\footnote{This problem will be revisited in a stochastic form in Section \ref{subsec_ce_e2e_problem} later, where the detailed system model and expression are specified.}

In this problem, both the sum-rate function and the power constraint are explicit functions of $\mathbf{H}_\mathrm{DL}$ and $\mathbf{V}$ and do not involve any underlying distribution. 
In this case, permuting the rows or columns of $\mathbf{H}_\mathrm{DL}$ is equivalent to reordering the system entities, leaving the objective and constraint invariant \cite{zhao2024understanding}. 
Consequently, the optimal precoding policy $\mathbf{V}=f(\mathbf{H}_\mathrm{DL})$ is PE.

However, this intuition does not hold for the policy $\mathbf{Y}=f(\mathbf{X})$ arising from the stochastic optimization problem \ref{prob_stochastic}. 
Specifically, relabeling antennas or users corresponds to permuting all context associated with these entities in the optimization problem. 
In deterministic problems like \textbf{P2}, the relevant context is fully captured by the measurements, so permuting $\mathbf{X}$ is equivalent to relabeling the entities. 
By contrast, in the stochastic problem \textbf{P1}, its optimal solution $\mathbf{Y}$ is determined not only by the observed measurements $\mathbf{X}$ but also by the conditional distribution $P_{\mathtt{Z}|\mathtt{X}}(\mathbf{Z}|\mathbf{X})$.
Therefore, unless the conditional distribution is PI, permuting $\mathbf{X}$ alone is not equivalent to relabeling the entities and the resulting policy is not PE.

The following proposition formalizes the condition under which the policy $\mathbf{Y} = f(\mathbf{X})$ exhibits PE property.

\begin{proposition}\label{proposition_pe}
$\mathbf{Y} = f(\mathbf{X})$ is PE to antennas or users if $U^\prime(\mathbf{Y}, \mathbf{Z})$, $\mathbf{c}(\mathbf{X}, \mathbf{Y})$ and $P_{\mathtt{Z}|\mathtt{X}}(\mathbf{Z}|\mathbf{X})$ in \ref{prob_stochastic} are PI to antennas or users.
\end{proposition}
\vspace{-4mm}
\begin{proof}
See Appendix \ref{app_proposition_pe}.
\end{proof}

Proposition \ref{proposition_pe} indicates that the PE property of the policy $\mathbf{Y} = f(\mathbf{X})$ depends on the PI property of $P_{\mathtt{Z}|\mathtt{X}}(\mathbf{Z} | \mathbf{X})$ due to the expectation in the objective function of \ref{prob_stochastic}.
We refer to such policies whose permutation property relies on the underlying distribution as \textit{distribution-dependent policies}.
In fact, the random variable may also appear in the constraints or in ways other than expectations (e.g., in probabilistic forms) \cite{kall1994stochastic}, while we do not elaborate on these scenarios for the sake of clarity.
By contrast, if neither the objectives nor the constraints of a wireless problem depend on unknown latent variables, such as the precoding problem \ref{prob_classic_precoding}, then the permutation property of the resulting policy is fixed, which is referred to as a \textit{distribution-independent} policy.

In what follows, we provide a toy example of parameter estimation to show how the underlying distribution affects the permutation property of the policy.

\begin{example}
Consider estimating $\mathbf{z} \in \mathbb{R}^{2 \times 1}$ from observation $\mathbf{x} \in \mathbb{R}^{2 \times 1}$, and the minimal MSE estimate is $\mathbf{y} = \mathbb{E}_{\mathtt{z} | \mathtt{x}}[\mathbf{z} | \mathbf{x}]$.
Denote the estimation policy as $\mathbf{y} = f(\mathbf{x})$.
Suppose $\mathbf{x}$ and $\mathbf{z}$ satisfy $\mathbf{x} = \mathbf{z} + \mathbf{n}$, where $\mathbf{n}$ represents noise independent of $\mathbf{z}$.
Denote $\mathbf{x} = [x_1, x_2]^\mathsf{T}$, $\mathbf{z} = [z_1, z_2]^\mathsf{T}$, $\mathbf{n} = [n_1, n_2]^\mathsf{T}$, and assume that 
$\mathbf{z} \sim \mathbb{N}(\mathbf{0}, \mathrm{diag}(\sigma_1^2, \sigma_2^2))$, $\mathbf{n} \sim \mathbb{N}(0, \sigma_n^2\mathbf{I})$.
Then, the resulting estimator is $\mathbf{y} = f(\mathbf{x}) = [\frac{\sigma_1^2}{\sigma_1^2 + \sigma_n^2}x_1, \frac{\sigma_2^2}{\sigma_2^2 + \sigma_n^2}x_2]^\mathsf{T}$.
If the observations are permuted to $\mathbf{x}^\prime = \boldsymbol{\Pi}\mathbf{x} = [x_2, x_1]^\mathsf{T}$, the corresponding optimal estimate will become $\mathbf{y}^\prime = f(\mathbf{x}^\prime) = [\frac{\sigma_1^2}{\sigma_1^2 + \sigma_n^2}x_2, \frac{\sigma_2^2}{\sigma_2^2 + \sigma_n^2}x_1]^\mathsf{T}$, which differs from $\boldsymbol{\Pi}\mathbf{y} = [\frac{\sigma_2^2}{\sigma_2^2 + \sigma_n^2}x_2, \frac{\sigma_1^2}{\sigma_1^2 + \sigma_n^2}x_1]^\mathsf{T}$ if $\sigma_1 \neq \sigma_2$. In other words, the policy exhibits PE property iff $\sigma_1 = \sigma_2$.
\end{example}

\subsection{Permutation Equivariant Policy Augmentation}\label{subsec_augmented_policy}

The policy $\mathbf{Y}=f(\mathbf{X})$ obtained from \ref{prob_stochastic} is not always PE, while incorporating permutation equivariance into the policy structure is highly desirable when learning the policy using DNNs, since it introduces an inductive biases that help reduce the hypothesis space and improve learning efficiency. 

To this end, we construct an augmented policy by incorporating additional features that encode the statistical distinction among entities, so that permutations can be explicitly represented in the input space. 

To help understand, we consider a single-user setting and focus only on the antenna dimension, while the extension to the user dimension is straightforward.
Without loss of generality, we set $D_x=D_y=1$ so that the tensors $\mathbf{X}, \mathbf{Y}, \mathbf{Z}$ reduce to vectors $\mathbf{x}, \mathbf{y}, \mathbf{z}\in\mathbb{C}^{N\times1}$.
Let $\mathtt{Z}=[\mathtt{z}_1,\ldots,\mathtt{z}_N]^\mathsf{T}$ denote the latent random variables and 
$\mathtt{X}=[\mathtt{x}_1,\ldots,\mathtt{x}_N]^\mathsf{T}$ denote the corresponding observations, with conditional PDF $P_{\mathtt{Z}|\mathtt{X}}(\mathbf{z}|\mathbf{x})$.

\subsubsection{Permutation Structure of Distribution Families}

Consider a permutation of indices represented by a permutation matrix $\boldsymbol{\Pi}$ and its associated permutation mapping $\pi(\cdot)$, and the permuted variables are $\mathtt{Z}^\pi =[\mathtt{z}_{\pi(1)},\ldots,\mathtt{z}_{\pi(N)}]^\mathsf{T}$ and $\mathtt{X}^\pi =[\mathtt{x}_{\pi(1)},\ldots,\mathtt{x}_{\pi(N)}]^\mathsf{T}$.
After permuting the indices of the antennas, the conditional distribution becomes $P_{\mathtt{Z}^\pi | \mathtt{X}^\pi}(\mathbf{z}|\mathbf{x})$.

Since these conditional PDFs correspond to the same underlying statistical model under different index orderings, they satisfy
\begin{equation}
P_{\mathtt{Z}|\mathtt{X}}(\mathbf{z}|\mathbf{x}) = P_{\mathtt{Z}^\pi | \mathtt{X}^\pi}(\boldsymbol{\Pi}^\mathsf{T} \mathbf{z}|\boldsymbol{\Pi}^\mathsf{T} \mathbf{x}),
\label{eq_pdf_permutation_relation}
\end{equation}
which indicates that the permutation property is reflected in the relationship among different conditional PDFs associated with different index orderings.

\subsubsection{Positional Embedding and Augmented Distribution}

To represent the permutation structure among different conditional PDFs using a single unified function, we introduce an additional vector $\mathbf{p}\in\mathbb{R}^{N\times1}$ whose elements are associated with the antennas. 
This vector aims to encode the antenna indices, and we refer to it as a \textit{positional embedding vector}.
By incorporating $\mathbf{p}$ into the policy input, we construct an augmented PDF $\tilde P(\mathbf{z}|\mathbf{x},\mathbf{p})$ defined as
\begin{equation}\label{eq_augmented_pdf_definition}
\tilde P(\mathbf{z}|\mathbf{x},\mathbf{p})=
\begin{cases}
P_{\mathtt{Z}|\mathtt{X}}(\mathbf{z}|\mathbf{x}), & \mathbf{p}=\tilde{\mathbf{p}},\\
P_{\mathtt{Z}^{\pi}|\mathtt{X}^{\pi}}(\mathbf{z}|\mathbf{x}), & \mathbf{p}=\boldsymbol{\Pi}^{\mathsf T}\tilde{\mathbf{p}},
\end{cases}
\end{equation}
where $\tilde{\mathbf{p}}\in\mathbb{R}^{N\times1}$ is a predetermined embedding vector associated with an initial index ordering. 
In $\tilde P(\mathbf{z}|\mathbf{x},\mathbf{p})$, the variable $\mathbf{p}$ takes values from the finite set $\{\boldsymbol{\Pi}^{\mathsf T}\tilde{\mathbf{p}}\ |\ \boldsymbol{\Pi}\ \text{is a permutation matrix}\}$, each element of which corresponds to a permutation of antenna indices and specifies the conditional distribution associated with that ordering.

The following proposition shows that such an augmented representation can always be constructed.

\begin{proposition}\label{prop_aug_pdf}
The augmented conditional distribution $\tilde P(\mathbf{z}|\mathbf{x},\mathbf{p})$ is well defined if for any two permutations $\boldsymbol{\Pi}_1$ and $\boldsymbol{\Pi}_2$, $\boldsymbol{\Pi}_1^\mathsf{T}\mathbf{p}=\boldsymbol{\Pi}_2^\mathsf{T}\mathbf{p}$ implies $P_{\mathtt{Z}^{\pi_1}|\mathtt{X}^{\pi_1}}(\mathbf{z}|\mathbf{x}) = P_{\mathtt{Z}^{\pi_2}|\mathtt{X}^{\pi_2}}(\mathbf{z}|\mathbf{x})$.
\end{proposition}

According to Proposition \ref{prop_aug_pdf}, the embedding vector is not unique.
A simple sufficient construction is to assign distinct values to its elements, e.g., $\tilde{\mathbf{p}}=[1,\ldots,N]^\mathsf{T}$.

Combining \eqref{eq_pdf_permutation_relation} and \eqref{eq_augmented_pdf_definition}, we obtain
\begin{equation}
\tilde P(\boldsymbol{\Pi}^\mathsf{T}\mathbf{z}\mid \boldsymbol{\Pi}^\mathsf{T}\mathbf{x},\boldsymbol{\Pi}^\mathsf{T}\mathbf{p})
=
\tilde P(\mathbf{z}|\mathbf{x},\mathbf{p}),
\end{equation}
which shows that the augmented conditional distribution $\tilde P(\mathbf{z}|\mathbf{x},\mathbf{p})$ is PI, regardless of the conditional PDF $P_{\mathtt{Z}|\mathtt{X}}(\mathbf{z}|\mathbf{x})$.

\subsubsection{Permutation Property of the Augmented Policy}

Based on the positional embedding, we construct an augmented policy by concatenating the embedding with the original input vector. 
Specifically, define the augmented input as $(\mathbf{x},\mathbf{p})$, and the augmented policy $\mathbf{y}=\tilde f(\mathbf{x},\mathbf{p})$ is defined as the mapping from $(\mathbf{x},\mathbf{p})$ to the optimal solution of the following problem,
\refstepcounter{problem}\label{prob_stochastic_augmented}
\begin{IEEEeqnarray}{rcl}\IEEEyesnumber\label{eq_general_optimization_augmented}
\theproblem: \quad \max_{\mathbf{y}} & \quad & U(\mathbf{x}, \mathbf{y}, \mathbf{p}) = \mathbb{E}\!\left[U^\prime(\mathbf{y}, \mathbf{z}) \mid \mathbf{x}, \mathbf{p}\right] \IEEEyessubnumber\label{eq_objective_general_augmented}\\
s.t. & \quad & \mathbf{c}(\mathbf{x}, \mathbf{y}) {\leq} \mathbf{0},\IEEEyessubnumber
\end{IEEEeqnarray}
where the expectation is taken with respect to the augmented conditional distribution $\tilde P(\mathbf{z}|\mathbf{x},\mathbf{p})$.

When $\mathbf{p}=\tilde{\mathbf{p}}$, \textbf{P3} is equivalent to \textbf{P1} in the sense that they share the same optimal solution. 
Moreover, the following proposition shows that introducing the positional embedding converts the original distribution-dependent policy into a distribution-independent policy.

\begin{proposition}\label{prop_aug_policy_pe}
$\mathbf{y} = \tilde{f}(\mathbf{x}, \mathbf{p})$ is PE to antennas or users if $U^\prime(\mathbf{y}, \mathbf{z})$ and $\mathbf{c}(\mathbf{x}, \mathbf{y})$ in \ref{prob_stochastic} are PI to antennas or users.
\end{proposition}

\begin{proof}
The result follows directly from Proposition \ref{proposition_pe} since $\tilde P(\mathbf{z}|\mathbf{x},\mathbf{p})$ is permutation invariant.
\end{proof}

\begin{remark}
For the multi-user case, an additional positional embedding can be introduced to encode the user indices.
\end{remark}

\section{Positional Attention-based GNN}\label{section_covgnn}

In this section, we propose a positional attention mechanism and integrate it with GNN to learn the augmented policy.
We begin by modelling the MU-MISO system as a heterogeneous graph and presenting an edge-GNN for learning policies on it.
Then, we analyze the permutation property of the edge-GNN, and show that standard attention mechanisms preserve permutation equivariance.
Finally, we incorporate the positional embedding vector into the attention coefficients, proposing a positional attention-based GNN whose permutation property is consistent with the augmented policy.

\subsection{Edge-GNN for Learning MU-MISO Policies}

We first introduce the graph representation and the basic edge-GNN architecture used to learn wireless policies.
GNNs operate on graphs $\mathcal{G}=(\mathcal{V},\mathcal{E})$, where $\mathcal{V}$ and $\mathcal{E}$ denote the sets of vertices and edges, respectively.
By enforcing parameter sharing across vertices and edges of the same type, the number of trainable parameters of GNNs is independent of the graph size, which is a premise for generalizing the learned policy to different problem sizes.

To support size generalization with respect to both the number of BS antennas and the number of users in the wireless policy $\mathbf{Y}=f(\mathbf{X})$, we construct a heterogeneous graph consisting of $N$ antenna vertices $\mathsf{A}_1,\ldots,\mathsf{A}_N$ and $K$ user vertices $\mathsf{U}_1,\ldots,\mathsf{U}_K$.
Since both the input features $\mathbf{X}\in\mathbb{C}^{N\times K\times D_x}$ and the output actions $\mathbf{Y}\in\mathbb{C}^{N\times K\times D_y}$ are naturally indexed by antenna-user pairs, it is straightforward to assign features and actions to the edges connecting antenna and user vertices.
Specifically, we denote the edge between $\mathsf{A}_n$ and $\mathsf{U}_k$ as $(n, k)$, and associate it with feature vector $\mathbf{X}[n, k, :]$ and action vector $\mathbf{Y}[n, k, :]$, respectively.
Consequently, the graph is a complete bipartite graph as shown in Fig.~\ref{fig_topology}.

\begin{figure}[htbp]
\centering
\includegraphics[width=0.7\linewidth]{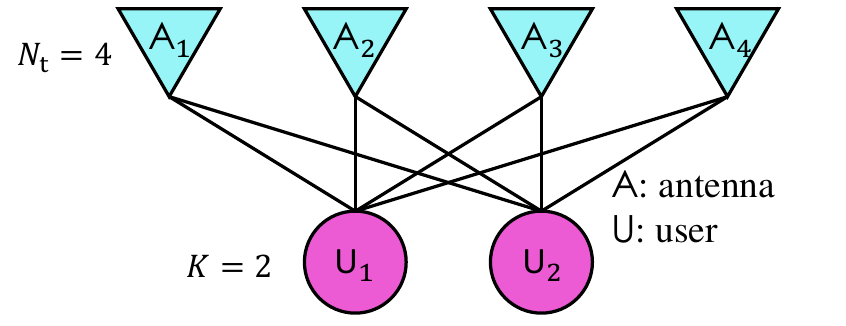}\vspace{-3mm}
\caption{Heterogeneous graph for learning policies in MU-MISO systems.}
\label{fig_topology}
\end{figure}

Denote the hidden representation of edge $(n, k)$ in the $l$-th layer as $\mathbf{d}^{(l)}_{n, k} \in \mathbb{R}^{C_l \times 1}$, where $C_l$ is the dimension of the hidden representation and is a hyper-parameter.
The hidden representations of all edges can be expressed as a tensor $\mathbf{D}^{(l)} \in \mathbb{R}^{N \times K \times C_l}$ with $\mathbf{D}^{(l)}[n, k, :] = \mathbf{d}^{(l)}_{n, k}$.
Denote $N_\mathrm{L}$ as the number of layers of the GNN.
The input and output of the edge-GNN are $\mathbf{d}_{n, k}^{(0)} = [\mathfrak{Re}(\mathbf{X}[n, k, :])^\mathsf{T}, \mathfrak{Im}(\mathbf{X}[n, k, :])^\mathsf{T}]^\mathsf{T} \in \mathbb{R}^{2D_x \times 1}$ and $\mathbf{d}_{n, k}^{(N_\mathrm{L})} = [\mathfrak{Re}(\mathbf{Y}[n, k, :])^\mathsf{T}, \mathfrak{Im}(\mathbf{Y}[n, k, :])^\mathsf{T}]^\mathsf{T} \in \mathbb{R}^{2D_y \times 1}$, which consist of the real and imaginary parts of the features and actions, respectively.

To improve expressive power, we incorporate an attention mechanism into the aggregation process and adopt vector-valued attention coefficients as in \cite{liu2024multidimensional}.
The update equation of the edge-GNN is given by
\begin{IEEEeqnarray*}{rcl}\label{eq_gat_update}
\mathbf{d}_{n,k}^{(l)}
&=&
f_\mathsf{C}^{(l)}\Bigg(
\mathbf{d}_{n,k}^{(l-1)},
\underbrace{\sum_{n^\prime=1}^{N}
\boldsymbol{\rho}_{n,n^\prime,k}^{(l)}
\odot
\mathring{\mathbf{d}}_{n^\prime,k}^{(l-1)}}_{\text{aggregate along the antenna dimension}},
\\
&&\qquad
\underbrace{\sum_{k^\prime=1}^{K}
\boldsymbol{\omega}_{k,k^\prime,n}^{(l)}
\odot
\mathring{\mathbf{d}}_{n,k^\prime}^{(l-1)}}_{\text{aggregate along the user dimension}}
\Bigg),\IEEEyesnumber
\end{IEEEeqnarray*}
where $\mathring{\mathbf{d}}_{n,k}^{(l-1)}=f_\mathsf{P}^{(l)}(\mathbf{d}_{n,k}^{(l-1)})$ denotes the processed hidden representation, and $f_\mathsf{C}^{(l)}(\cdot)$ and $f_\mathsf{P}^{(l)}(\cdot)$ are respectively the combination and processing functions implemented by FNNs.
The vectors $\boldsymbol{\rho}_{n,n^\prime,k}^{(l)}$ and $\boldsymbol{\omega}_{k,k^\prime,n}^{(l)}$ denote attention coefficients along the antenna and user dimensions.

As shown in \eqref{eq_gat_update}, the same functions $f_\mathsf{C}^{(l)}(\cdot)$ and $f_\mathsf{P}^{(l)}(\cdot)$ are used to update the hidden representations for all edges.
Therefore, the number of free parameters of the edge-GNN is independent of $N$ and $K$, which is necessary for generalization to other problem sizes.

\subsection{Positional Attention Mechanism}

To learn the augmented policy constructed in the previous section, we first analyze the PE property of the edge-GNN, and then propose a positional attention mechanism to incorporate the positional embedding.

As in Section \ref{subsec_augmented_policy}, we still focus on the antenna dimension by considering a single-user scenario and omitting the subscript $k$ to highlight the key idea of our design.
Moreover, for notational simplicity, we set $C_l=1$ so that the hidden representations reduce to scalars.
Under these simplifications, the update equation in \eqref{eq_gat_update} becomes
\begin{equation}
d_{n}^{(l)} = f_\mathsf{C}^{(l)} \Big(d_{n}^{(l-1)}, \sum_{n^\prime=1}^{N}\rho_{n,n^\prime}^{(l)} \mathring{d}_{n^\prime}^{(l-1)}\Big).
\label{eq_update_pagnn_simplified_new}
\end{equation}

\subsubsection{Permutation Property of Attention-based GNN}

Since the choices of $f_\mathsf{C}^{(l)}(\cdot)$ and $f_\mathsf{P}^{(l)}(\cdot)$ do not affect the permutation property of the GNN layer \cite{gilmer2017neural}, we adopt $f_\mathsf{C}^{(l)}(x)=\mathrm{sum}(\cdot)$ and $f_\mathsf{P}^{(l)}(x)=x$ to further simplify the update equation in \eqref{eq_update_pagnn_simplified_new} and highlight the role of the attention coefficients, yielding
\begin{equation}
d_n^{(l)} = d_n^{(l - 1)} + \sum_{n^\prime = 1}^{N} \rho_{n, n^\prime}^{(l)} d_{n^\prime}^{(l - 1)}.
\end{equation}

By stacking $d_n^{(l)}, n=1,\ldots,N$ into a vector $\mathbf{d}^{(l)} \triangleq [d_1^{(l)}, \ldots, d_{N}^{(l)}]^\mathsf{T}$, the update equation can be written in matrix form as
\begin{equation}\label{eq_update_mat}
\mathbf{d}^{(l)} = f^{(l)}(\mathbf{d}^{(l - 1)}, \mathbf{A}^{(l)}) \triangleq (\mathbf{A}^{(l)} + \mathbf{I}_{N}) \mathbf{d}^{(l - 1)},
\end{equation}
where $\mathbf{A}^{(l)} \in \mathbb{R}^{N \times N}$ is the \textit{attention matrix} with entries $[\mathbf{A}^{(l)}]_{n, n^\prime} \triangleq \rho_{n, n^\prime}^{(l)}$.
It can be easily verified that the update equation satisfies
\begin{equation}\label{eq_att_permutation}
\boldsymbol{\Pi}^\mathsf{T}\mathbf{d}^{(l)} = f^{(l)}(\boldsymbol{\Pi}^\mathsf{T}\mathbf{d}^{(l - 1)}, \boldsymbol{\Pi}^\mathsf{T}\mathbf{A}^{(l)}\boldsymbol{\Pi}).
\end{equation}

For standard attention mechanisms, the attention coefficients are determined by the hidden representations from the previous layer $\mathbf{d}^{(l - 1)}$ or the input feature $\mathbf{d}^{(0)}$. 
As an example, the dot-product attention mechanism in \cite{vaswani2017attention} computes the coefficients as $\rho_{n, n^\prime}^{(l)} = \rho^{(l)}(d_{n}^{(l - 1)}, d_{n^\prime}^{(l - 1)}) \triangleq w_\mathsf{Q}^{(l)}d_{n}^{(l - 1)} \cdot w_\mathsf{K}^{(l)}d_{n^\prime}^{(l - 1)} / \sqrt{C_l}$ with trainable parameters $w_\mathsf{Q}^{(l)}$ and $w_\mathsf{K}^{(l)}$.
The resulting attention matrix can be expressed as
\begin{equation}\label{eq_att_mat_existing}
\begin{small}
\mathbf{A}^{(l)}(\mathbf{d}^{(l - 1)})\triangleq
\begin{bmatrix}
\rho^{(l)}(d^{(l - 1)}_1, d^{(l - 1)}_1) & \cdots & \rho^{(l)}(d^{(l - 1)}_1, d^{(l - 1)}_{N}) \\
\vdots & \ddots & \vdots \\
\rho^{(l)}(d^{(l - 1)}_{N}, d^{(l - 1)}_1) & \cdots & \rho^{(l)}(d^{(l - 1)}_{N}, d^{(l - 1)}_{N})
\end{bmatrix}.
\end{small}
\end{equation}

From \eqref{eq_att_mat_existing}, permuting the hidden representations to $\boldsymbol{\Pi}^\mathsf{T} \mathbf{d}^{(l - 1)}$ leads to a corresponding permutation of the attention matrix, i.e., $\mathbf{A}^{(l)}(\boldsymbol{\Pi}^\mathsf{T} \mathbf{d}^{(l - 1)}) = \boldsymbol{\Pi}^\mathsf{T} \mathbf{A}^{(l)}(\mathbf{d}^{(l - 1)}) \boldsymbol{\Pi}$, and the updated output becomes
\begin{IEEEeqnarray*}{rcl}
\mathbf{d}^{\prime (l)} 
&=& f^{(l)}(\boldsymbol{\Pi}^\mathsf{T} \mathbf{d}^{(l - 1)}, \mathbf{A}^{(l)}(\boldsymbol{\Pi}^\mathsf{T} \mathbf{d}^{(l - 1)})) \\
&=& f^{(l)}(\boldsymbol{\Pi}^\mathsf{T} \mathbf{d}^{(l - 1)}, \boldsymbol{\Pi}^\mathsf{T} \mathbf{A}^{(l)}(\mathbf{d}^{(l - 1)}) \boldsymbol{\Pi}) = \boldsymbol{\Pi}^\mathsf{T} \mathbf{d}^{(l)},
\end{IEEEeqnarray*}
Therefore, the individual GNN layer and the complete mapping learned by stacking $N_\mathrm{L}$ such layers are PE to antennas.

\subsubsection{Injecting Positional Embeddings}

Since the distribution-dependent policy may not exhibit PE property, employing the edge-GNN to learn it directly may introduce a mismatched inductive bias.
To address this issue, we incorporate the positional embedding into the edge-GNN and employ it to learn the augmented policy $\mathbf{y}=\tilde f(\mathbf{x},\mathbf{p})$ from \textbf{P3}.

A straightforward approach is to concatenate the positional embedding vector $\tilde{\mathbf{p}}$ with the input features, i.e., expanding $\mathbf{d}_{n}^{(0)}=[\mathfrak{Re}(x_n)^\mathsf{T},\mathfrak{Im}(x_n)^\mathsf{T}]^\mathsf{T}$ to $\mathbf{d}_{n}^{(0)}=[\mathfrak{Re}(x_n)^\mathsf{T},\mathfrak{Im}(x_n)^\mathsf{T},\tilde{p}_n]^\mathsf{T}$.
This is similar to the positional encoding used in Transformers \cite{vaswani2017attention}. 
However, this increases the feature dimensionality and makes the positional information difficult for subsequent layers to capture.

Instead, we incorporate the positional embeddings into the attention mechanism so that positional information affects the aggregation at every layer. 
Since the augmented policy satisfies $\boldsymbol{\Pi}^\mathsf{T}\mathbf{y}=\tilde f(\boldsymbol{\Pi}^\mathsf{T}\mathbf{x},\boldsymbol{\Pi}^\mathsf{T}\mathbf{p})$, each GNN layer should preserve the same permutation property. 
Specifically, with the attention matrix $\mathbf{A}^{(l)}(\tilde{\mathbf{p}})$, the $l$-th layer satisfies
\begin{equation}
\boldsymbol{\Pi}^\mathsf{T}\mathbf{d}^{(l)} = f^{(l)}(\boldsymbol{\Pi}^\mathsf{T}\mathbf{d}^{(l-1)},\mathbf{A}^{(l)}(\boldsymbol{\Pi}^\mathsf{T}\tilde{\mathbf{p}})).
\end{equation}

Comparing it with \eqref{eq_att_permutation}, the attention matrix induced by the positional embedding should satisfy
\[
\mathbf{A}^{(l)}(\boldsymbol{\Pi}^\mathsf{T}\tilde{\mathbf{p}})
=
\boldsymbol{\Pi}^\mathsf{T}\mathbf{A}^{(l)}(\tilde{\mathbf{p}})\boldsymbol{\Pi}.
\]
This requirement can be satisfied by constructing the attention coefficients based on pairwise positional embeddings. Specifically, we define
\begin{equation}\label{eq_att_coef}
\rho_{n,n'}^{(l)} = f_\mathsf{A}^{(l)}\!\left(\mathrm{EB}(\tilde p_n,\tilde p_{n'})\right),
\end{equation}
where $\mathrm{EB}(\cdot,\cdot)$ is a deterministic pairwise embedding function that models the interaction between the positional embeddings of two antennas.

In practice, the mapping from the antenna index $n$ to its corresponding embedding value $\tilde{p}_n$ is fixed and can be integrated into the embedding function $\mathrm{EB}(\cdot, \cdot)$. 
By setting $\tilde p_n = n$ without loss of generality, \eqref{eq_att_coef} can be simplified as
\[
\rho_{n,n'}^{(l)} = f_\mathsf{A}^{(l)}\!\left(\mathrm{EB}(n,n')\right).
\]


To illustrate the role of $\mathrm{EB}(\cdot, \cdot)$, consider a typical scenario where antennas are arranged in structured arrays such as ULAs or UPAs.
In such systems, the channel statistics associated with different antennas are determined by their spatial locations. 
Accordingly, the function $\mathrm{EB}(n, n^\prime)$ can be designed to reflect the relative spatial relationship of the $n$-th and the $n^\prime$-th antennas, which facilitates the GNN to capture the underlying channel correlation structure.

\subsubsection{Update Equation}

The positional attention mechanism was derived under the simplified single-user setting in \eqref{eq_update_pagnn_simplified_new}. 
We now return to the general MU-MISO setting and provide the update equation of the positional attention-based GNN (referred to as \textit{PA-GNN}), which is
\begin{IEEEeqnarray*}{rcl}
\mathbf{d}_{n, k}^{(l)} & = & f_\mathsf{C}^{(l)}\Big({\mathbf{d}}_{n, k}^{(l - 1)}, \sum_{n^\prime = 1}^{N}f_\mathsf{A}^{(l)}(\mathrm{EB}(n, n^\prime)) \odot \mathring{\mathbf{d}}_{n^\prime, k}^{(l - 1)}, \\
& & \quad \sum_{k^\prime = 1}^{K}\boldsymbol{\omega}_{k, k^\prime, n}^{(l)} \odot \mathring{\mathbf{d}}_{n, k^\prime}^{(l-1)}\Big), \IEEEyesnumber\label{eq_update_pgnn}
\end{IEEEeqnarray*}
where the hidden representations are restored to vectors with dimension $C_l>1$, and the output of $f_\mathsf{A}^{(l)}(\cdot)$ has the same dimension.
Here, the positional attention mechanism is applied only along the antenna dimension, since the two problems considered in the next section are PE with respect to users.
Nonetheless, it can be extended to the user dimension if needed.

The number of trainable parameters in PA-GNN is still independent of $N$. 
This is because $f_\mathsf{A}^{(l)}(\cdot)$ is shared and produces all attention coefficients in $\mathbf{A}^{(l)}$, whose input and output dimensions are fixed and do not scale with $N$.

\begin{remark}
The proposed positional attention mechanism is not limited to the bipartite graph illustrated in Fig.~\ref{fig_topology}, and can be extended to other graph structures.
For example, when applying GNNs to learn policies for MU-MIMO systems over more complex graph topologies \cite{liu2024multidimensional}, the relative positions of user antennas can also be leveraged.
\end{remark}

\section{Learning Channel Estimation and E2E Precoding}\label{section_ce_and_e2e}

In this section, we focus on two practical distribution-dependent policies, i.e., the channel estimation policy and E2E precoding policy, and employ the proposed PA-GNN to learn them.
We first prove that both policies are PE to users, but not to antennas when the channels are spatially correlated.
Then, we design the embedding function of the PA-GNN for learning the two policies, by investigating the structure of the spatial channel correlation matrix to gain useful insights.

\subsection{Channel Estimation and E2E Precoding Problems}\label{subsec_ce_e2e_problem}

Consider a TDD system, whose frame structure is depicted in Fig.~\ref{fig_frame_structure}.
Each frame consists of two phases, i.e., the uplink phase with $L_\mathrm{UL}$ subframes and the downlink phase with $L_\mathrm{DL}$ subframes.
The duration of each subframe is $T_\mathrm{s}$.
\begin{figure}[htbp]
\centering
\includegraphics[width=1\linewidth]{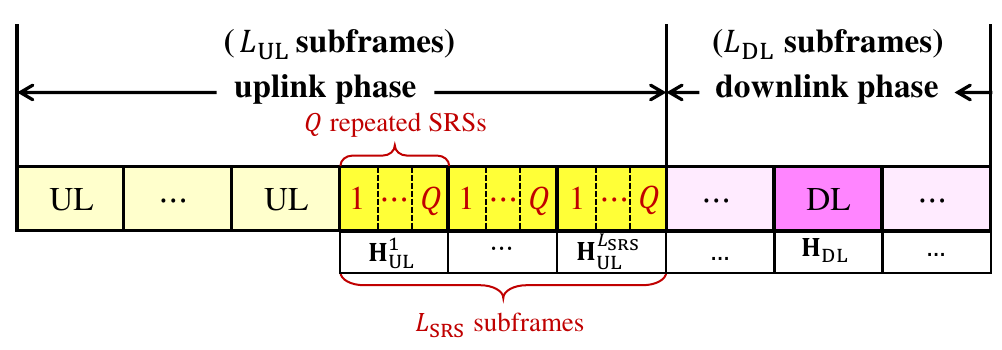}\vspace{-3mm}
\caption{The frame structure of the TDD system.}
\label{fig_frame_structure}
\end{figure}

\subsubsection{Channel Estimation Problem}

The first problem considers uplink channel estimation.

At the end of the uplink phase, $L_\mathrm{SRS}$ subframes are used for sending SRSs, where each subframe consists of $Q$ repeated SRS blocks to improve the signal-to-noise ratio (SNR) for channel acquisition.
In each block, the users send orthogonal SRSs with length $K$ to the BS simultaneously.

Denote the channel vector from the $k$-th user to the BS in the $l$-th subframe for SRS transmission as $\mathbf{h}^{l}_{\mathrm{UL}, k} \in \mathbb{C}^{N \times 1}, l=1,\cdots, L_\mathrm{SRS}$, and denote the uplink channel matrix as $\mathbf{H}_\mathrm{UL} \in \mathbb{C}^{N \times K \times L_\mathrm{SRS}}$, where $\mathbf{H}_\mathrm{UL}[:, k, l] \triangleq \mathbf{h}_{\mathrm{UL}, k}^l$.
Since the SRSs among different users are orthogonal, the effective received SRSs at the BS in the frame are \cite{wang2024learning}
\begin{equation}\label{eq_avg_effective_srs}
\mathbf{R} = \mathbf{H}_\mathrm{UL} +\mathbf{N},
\end{equation}
where $\mathbf{N} \in \mathbb{C}^{N \times K \times L_\mathrm{SRS}}$ is white Gaussian noise with zero mean and variance $\frac{\sigma_\mathrm{UL}^2}{P_\mathrm{UL} Q}$, $P_\mathrm{UL}$ is the uplink transmit power of each user, and $\sigma_\mathrm{UL}^2$ is the uplink noise power.
In fact, $\mathbf{R} \in \mathbb{C}^{N \times K \times L_\mathrm{SRS}}$ is the least squares (LS) estimate of $\mathbf{H}_\mathrm{UL}$ \cite{hampton2013introduction}.

The uplink channels can be estimated from the effective received SRSs by minimizing the MSE, i.e.,
\refstepcounter{problem}\label{prob_ce}
\begin{IEEEeqnarray}{rcl}\label{eq_channel_estimation_problem}
\theproblem: \quad \min_{\widehat{\mathbf{H}}_\mathrm{UL}} & \quad & \mathbb{E}_{\mathtt{H}_\mathrm{UL} | \mathtt{R}}\left[U_\mathrm{MSE}(\mathbf{H}_\mathrm{UL}, \widehat{\mathbf{H}}_\mathrm{UL}) | \mathbf{R}\right],
\end{IEEEeqnarray}
where $\widehat{\mathbf{H}}_\mathrm{UL} \in \mathbb{C}^{N \times K \times L_\mathrm{SRS}}$ is the estimated uplink channel matrix, $U_\mathrm{MSE}(\mathbf{H}_\mathrm{UL}, \widehat{\mathbf{H}}_\mathrm{UL}) \triangleq ||\mathbf{H}_\mathrm{UL} - \widehat{\mathbf{H}}_\mathrm{UL}||_\mathrm{F}^2$, and the expectation is taken over $\mathbf{H}_\mathrm{UL}$ given $\mathbf{R}$.

The channel estimation policy is the mapping from the effective SRSs to the estimated channels, denoted as $\widehat{\mathbf{H}}_\mathrm{UL} = f_\mathrm{CE}(\mathbf{R})$.

\subsubsection{E2E Precoding Problem}

The second problem under consideration is to optimize the downlink precoders in an E2E manner to maximize the downlink sum rate.

According to the analyses in \cite{wang2024learning}, the optimization of the precoders in the $L_\mathrm{DL}$ downlink subframes can be decomposed into $L_\mathrm{DL}$ sub-problems.
Therefore, we optimize the precoder in a single downlink subframe, and the sum rate in this subframe is
\begin{equation}\label{eq_sum_rate}
R(\mathbf{H}_\mathrm{DL}, \mathbf{V})
 \triangleq \sum_{k=1}^K \log_2 \left(1 + \frac{|(\mathbf{h}_{\mathrm{DL}, k})^\mathsf{H} \mathbf{v}_k|^2}{\sum_{j \neq k}|(\mathbf{h}_{\mathrm{DL}, k})^\mathsf{H} \mathbf{v}_j|^2 + \sigma_\mathrm{DL}^2}\right),
\end{equation}
\noindent where $\mathbf{h}_{\mathrm{DL}, k} \in \mathbb{C}^{N \times 1}$ is the channel vector from the BS to the $k$-th user, $\mathbf{v}_k \in \mathbb{C}^{N \times 1}$ is the precoder for the $k$-th user, $\sigma_\mathrm{DL}^2$ is the downlink noise power, $\mathbf{H}_\mathrm{DL} = [\mathbf{h}_{\mathrm{DL}, 1}, \cdots, \mathbf{h}_{\mathrm{DL}, K}] \in \mathbb{C}^{N \times K}$ and $\mathbf{V} \triangleq [\mathbf{v}_1, \cdots, \mathbf{v}_K] \in \mathbb{C}^{N \times K}$ are the downlink channel matrix and the precoding matrix, respectively.

To enhance system performance, we optimize the downlink precoding matrix in an E2E manner.
Since the downlink channel is unknown, the precoding matrix is optimized to maximize the \textit{average} sum rate, i.e.,
\refstepcounter{problem}\label{prob_e2e}
\begin{IEEEeqnarray}{rcl}\IEEEyesnumber\label{eq_e2e_optimization}
\theproblem: \quad \max_{\mathbf{V}} & \quad & \mathbb{E}_{\mathtt{H}_\mathrm{DL} | \mathtt{R}}\left[R\left(\mathbf{H}_\mathrm{DL}, \mathbf{V}\right)|\mathbf{R}\right] \IEEEyessubnumber\\
s.t. & \quad & ||\mathbf{V}||_\mathrm{F}^2 \leq P_\mathrm{DL},\IEEEyessubnumber\label{eq_e2e_power_constraint}
\end{IEEEeqnarray}
where the expectation is taken over $\mathbf{H}_\mathrm{DL}$ given $\mathbf{R}$.
In fact, \textbf{P5} is the stochastic version of \textbf{P2} after considering the randomness of the downlink channel.

The E2E precoding policy is the mapping from the effective SRSs to the optimized precoder, denoted as $\mathbf{V} = f_\mathrm{E2E}(\mathbf{R})$.

\subsection{Permutation Properties of $f_\mathrm{CE}(\cdot)$ and $f_\mathrm{E2E}(\cdot)$}\label{subsec_pe_ce_e2e}

Problems \ref{prob_ce} and \ref{prob_e2e} are special cases of \ref{prob_stochastic}, where the associated objective and constraint functions, i.e., $U_\mathrm{MSE}(\mathbf{H}_\mathrm{UL}, \widehat{\mathbf{H}}_\mathrm{UL})$, $R\left(\mathbf{H}_\mathrm{DL}, \mathbf{V}\right)$ and $c(\mathbf{R}, \mathbf{V}) \triangleq ||\mathbf{V}||_\mathrm{F}^2 - P_\mathrm{DL}$, are PI to both antennas and users.
By Proposition \ref{proposition_pe}, the induced policies $f_\mathrm{CE}(\cdot)$ and $f_\mathrm{E2E}(\cdot)$ are PE if $P_{\mathtt{H}_\mathrm{UL} | \mathtt{R}}(\mathbf{H}_\mathrm{UL} | \mathbf{R})$ and $P_{\mathtt{H}_\mathrm{DL} | \mathtt{R}}(\mathbf{H}_\mathrm{DL} | \mathbf{R})$ are PI.
In what follows, we investigate the PI properties of these distributions to users and antennas, assuming $L_\mathrm{SRS} = 1$ for simplicity.
Then, the three-order tensors $\mathbf{R}$, $\mathbf{N}$ and $\mathbf{H}_\mathrm{UL}$ degenerate into matrices.

\subsubsection{PE Property to Users}

Denote $\mathbf{r}_k$ as the $k$-th column of $\mathbf{R}$, which represents the effective SRS of the $k$-th user.
Under the assumption that different users are spatially separated and statistically identical \cite{gesbert2007shifting}, their channels are independent and identically distributed.
Further considering that the receive noises $\mathbf{N}$ are white Gaussian, it is not hard to show that $P_{\mathtt{H}_\mathrm{UL} | \mathtt{R}}(\mathbf{H}_\mathrm{UL} | \mathbf{R}) = \prod_{k=1}^K P_{\mathtt{h}_{\mathrm{UL}} | \mathtt{r}}(\mathbf{h}_{\mathrm{UL}, k} | \mathbf{r}_k)$ and $P_{\mathtt{H}_\mathrm{DL} | \mathtt{R}}(\mathbf{H}_\mathrm{DL} | \mathbf{R}) = \prod_{k=1}^K P_{\mathtt{h}_{\mathrm{DL}} | \mathtt{r}}(\mathbf{h}_{\mathrm{DL}, k} | \mathbf{r}_k)$.
Therefore, applying any permutation matrix $\boldsymbol{\Pi}_\mathrm{U}$ to the channel matrices and the received SRS merely reorders the factors in the product, which does not change its value due to the commutativity of multiplication.
Hence, $P_{\mathtt{H}_\mathrm{UL} | \mathtt{R}}(\mathbf{H}_\mathrm{UL} | \mathbf{R})$ and $P_{\mathtt{H}_\mathrm{DL} | \mathtt{R}}(\mathbf{H}_\mathrm{DL} | \mathbf{R})$ are PI to users.

According to Proposition \ref{proposition_pe}, the channel estimation policy $f_\mathrm{CE}(\cdot)$ and the E2E precoding policy $f_\mathrm{E2E}(\cdot)$ are PE to users.

\subsubsection{PE Property to Antennas}

Different from the user dimension, the channel coefficients among different antennas may be spatially correlated, and thus the conditional distributions can not be decoupled.

For notational simplicity, we consider the single-user case since the number of users does not affect the permutation property to antennas.
Then, $\mathbf{H}_\mathrm{UL}$ and $\mathbf{R}$ degenerate into vectors $\mathbf{h}_\mathrm{UL}$ and $\mathbf{r}$, respectively, where the superscript $k$ is omitted.
The following proposition provides the conditions for the distributions being PI to antennas.

\begin{proposition}\label{prop_pe_ant}
Consider the following two conditions:
\begin{itemize}
\item[] (a) ${P}_{\mathtt{h}_\mathrm{UL}}(\cdot)$ is PI to antennas.
\item[] (b) $P_{\mathtt{h}_\mathrm{DL} | \mathtt{h}_\mathrm{UL}}(\cdot | \cdot)$ is PI to antennas.
\end{itemize}

If (a) holds, then $P_{\mathtt{h}_\mathrm{UL} | \mathtt{r}}(\mathbf{h}_\mathrm{UL} | \mathbf{r})$ is PI to antennas.

If both (a) and (b) hold, then $P_{\mathtt{h}_\mathrm{DL} | \mathtt{R}}(\mathbf{h}_\mathrm{DL} | \mathbf{R})$ is PI to antennas.
\end{proposition}
\begin{proof}
See in Appendix \ref{proof_prop_pe_ant}.
\end{proof}

Proposition \ref{prop_pe_ant} indicates that the PI property of ${P}_{\mathtt{h}_\mathrm{UL}}(\cdot)$ is a premise for that of $P_{\mathtt{h}_\mathrm{UL} | \mathtt{r}}(\mathbf{h}_\mathrm{UL} | \mathbf{r}), P_{\mathtt{h}_\mathrm{DL} | \mathtt{R}}(\mathbf{h}_\mathrm{DL} | \mathbf{R})$.
However, the condition does not hold in spatially correlated channels.

To see this, we consider the covariance matrix $\mathbf{C} \triangleq \mathbb{E}_{\mathtt{h}_\mathrm{UL}}[\mathbf{h}_\mathrm{UL} (\mathbf{h}_\mathrm{UL})^\mathsf{H}]$ of the channel vector.
It is not hard to prove that if ${P}_{\mathtt{h}_\mathrm{UL}}(\mathbf{h}_\mathrm{UL}) = {P}_{\mathtt{h}_\mathrm{UL}}(\boldsymbol{\Pi}_\mathrm{A}^\mathsf{T}\mathbf{h}_\mathrm{UL})$ holds for any permutation matrix $\boldsymbol{\Pi}_\mathrm{A}$, then all elements in $\mathbf{h}_\mathrm{UL}$ are identically distributed and equally correlated.
As a result, $\mathbf{C}$ exhibits the structure given in Fig.~\ref{fig_cor_mat_structure}(a) and $\mathbf{C} = \boldsymbol{\Pi}_\mathrm{A}^\mathsf{T} \mathbf{C} \boldsymbol{\Pi}_\mathrm{A}$ holds for any permutation matrices $\boldsymbol{\Pi}_\mathrm{A}$.

\begin{figure}[htbp]
\centering
\includegraphics[width=0.7\linewidth]{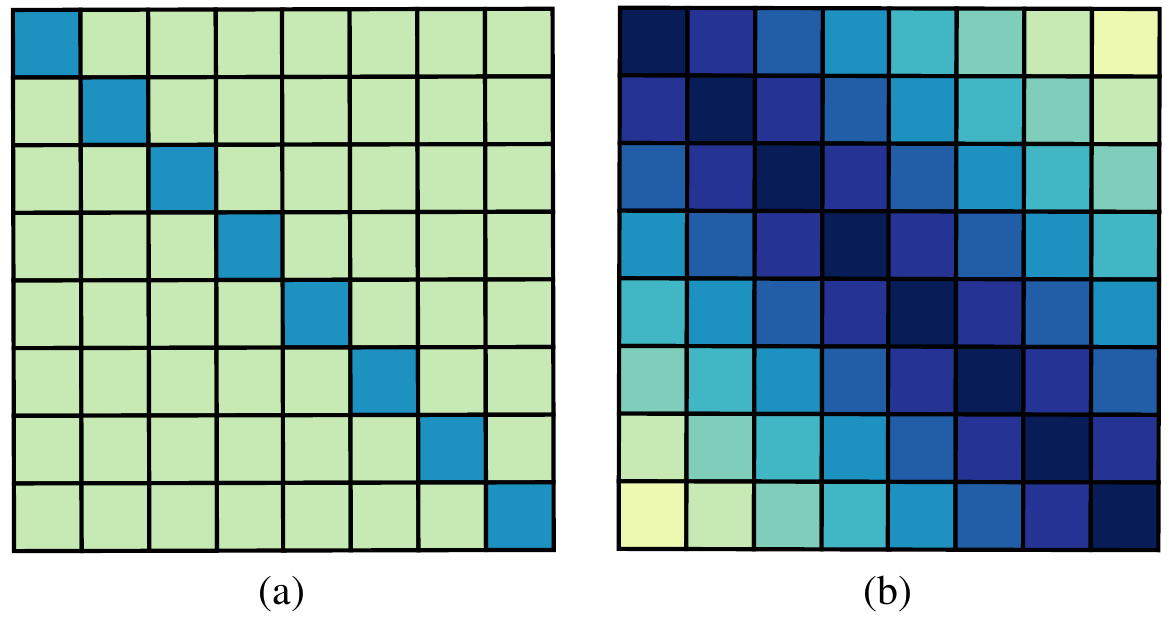}\vspace{-3mm}
\caption{Visualization of (a) symmetric matrices with identical diagonal and off-diagonal elements, and (b) Toeplitz matrices.}
\label{fig_cor_mat_structure}
\end{figure}

We now analyze the structure of covariance matrix of spatially correlated channel, where the correlation is induced by the closely spaced antennas.
Denote the steering vector of the antenna array equipped at the BS as $\mathbf{a}(\phi) \in \mathbb{C}^{N \times 1}$, where $\phi$ is the angle of departure (AoD), then the channel covariance matrix is \cite{liang2001downlink}
\begin{IEEEeqnarray*}{rl}
\mathbf{C} & = \mathbb{E}_\phi\left[\mathbf{a}(\phi) \mathbf{a}^\mathsf{H}(\phi)\right] \triangleq \mathbb{E}_\phi\left[\widetilde{\mathbf{C}}(\phi)\right].\IEEEyesnumber
\label{eq_covariance_matrix}
\end{IEEEeqnarray*}
The values of elements $\tilde{c}_{n, n^\prime}(\phi)$ in $\widetilde{\mathbf{C}}(\phi)$ depend on the relative position of the $n$-th and the $n^\prime$-th antennas, which is a function of the indices $n$ and $n^\prime$.
This makes $\mathbf{C}$ a structural matrix.

For example, if the BS is equipped with a ULA, then the steering vector is $\mathbf{a}_\mathrm{ULA}(\phi) \triangleq [1, \cdots, e^{j 2 \pi (N - 1) \frac{d \sin \phi}{\lambda}}]^\mathsf{T}$, where $d$ is the antenna spacing and $\lambda$ is the wavelength, and $\tilde{c}_{n, n^\prime}(\phi) = e^{j 2\pi (n - n^\prime) \frac{d \sin \phi}{\lambda}}$ depends on the relative position of the $n$-th and $n^\prime$-th antennas (i.e.,  $n - n^\prime$).
As a result, $\widetilde{\mathbf{C}}(\phi)$ and $\mathbf{C}$ are Toeplitz matrices \cite{liang2001downlink}.
In Fig.~\ref{fig_cor_mat_structure}(b), we visualize its structure for comparison, from which we can see that $\mathbf{C}$ for ULAs does not exhibit the structure given in Fig.~\ref{fig_cor_mat_structure}(a).
Similarly, for uniform planar arrays, the corresponding channel covariance matrix is \textit{block-Toeplitz} \cite{chen1993two}, and the elements within it depend on the relative positions of antenna pairs along horizontal and vertical directions.
Only if $d \gg \lambda$ such that the channels between different antennas are approximately \textit{uncorrelated}, the channel covariance matrix will degenerate into a scaled identity matrix with the structure in Fig.~\ref{fig_cor_mat_structure}(a).

According to Propositions \ref{proposition_pe} and \ref{prop_pe_ant}, the conditions under which the channel estimation policy $f_\mathrm{CE}(\cdot)$ and E2E precoding policy $f_\mathrm{E2E}(\cdot)$ are PE are not satisfied in spatially correlated channels.
Moreover, with similar derivations, they are also not guaranteed to be shift-equivariant to antennas in such scenarios.
Therefore, GNNs and CNNs that are PE or shift equivariant to antennas introduce policy-mismatched inductive biases.

As a validation and to help understand the analyses, we examine the PE property of the well-known LMMSE estimator, which has a closed-form expression.
From \eqref{eq_avg_effective_srs}, the LMMSE channel estimate of a single user is \cite{kay1993fundamentals}
\begin{equation}\label{eq_lmmse}
\hat{\mathbf{h}}_\mathrm{UL} = \mathbf{W}^\star\mathbf{r} \triangleq \mathbf{C}(\mathbf{C} + \frac{\sigma_\mathrm{UL}^2}{P_\mathrm{UL} Q}\mathbf{I}_{N})^{-1}\mathbf{r}.
\end{equation}

If the effective SRSs are permuted to $\mathbf{r}^\prime = \boldsymbol{\Pi}^\mathsf{T}_\mathrm{A} \mathbf{r}$, then $\hat{\mathbf{h}}_\mathrm{UL}^\prime = \mathbf{W}^\star\boldsymbol{\Pi}_\mathrm{A}^\mathsf{T}\mathbf{r}$.
If $\mathbf{C}$ has the structure given in Fig.~\ref{fig_cor_mat_structure}(a), it can be proved from \eqref{eq_lmmse} that $\mathbf{W}^\star$ will have the same structure. Then, according to Lemma 3 in \cite{zaheer2017deep}, $\mathbf{W}^\star\boldsymbol{\Pi}_\mathrm{A}^\mathsf{T}=\boldsymbol{\Pi}_\mathrm{A}^\mathsf{T}\mathbf{W}^\star$ holds and we have $\hat{\mathbf{h}}_\mathrm{UL}^\prime = \boldsymbol{\Pi}_\mathrm{A}^\mathsf{T} \hat{\mathbf{h}}_\mathrm{UL}$, which indicates that the LMMSE estimator is PE to antennas.
On the contrary, if $\mathbf{C}$ does not have the structure in Fig.~\ref{fig_cor_mat_structure} (a), then $\mathbf{W}^\star\boldsymbol{\Pi}_\mathrm{A}^\mathsf{T} \neq \boldsymbol{\Pi}_\mathrm{A}^\mathsf{T}\mathbf{W}^\star$ and thus $\hat{\mathbf{h}}_\mathrm{UL}^\prime \neq \boldsymbol{\Pi}_\mathrm{A}^\mathsf{T} \hat{\mathbf{h}}_\mathrm{UL}$, i.e., the LMMSE channel estimator is no longer PE to antennas due to spatial correlation.

\subsection{PA-GNN For Learning $f_\mathrm{CE}(\cdot)$ and $f_\mathrm{E2E}(\cdot)$}

We use the PA-GNN with the update equation in \eqref{eq_update_pgnn} to learn both the channel estimation and E2E precoding policies.
For both tasks, the input of the PA-GNN is the effective SRSs, while the output is either the estimated channels for learning $f_\mathrm{CE}(\cdot)$ or the precoding matrix for learning $f_\mathrm{E2E}(\cdot)$.
Moreover, an additional normalization step is applied for E2E precoding to satisfy the power constraint in \eqref{eq_e2e_power_constraint}.
The same embedding function $\mathrm{EB}(n, n^\prime)$ is used in both cases, while the attention coefficients $\boldsymbol{\omega}_{k, k^\prime, n}^{(l)}$ for aggregating user vertices are different, since the channel estimation can be decoupled among users.

\subsubsection{Embedding Function}

Since the attention coefficients are intended to capture the correlations of hidden representations of antennas, we resort to the spatial channel covariance matrix to guide the design the embedding function.

As analyzed in Section \ref{subsec_pe_ce_e2e}, the structure of the covariance matrix comes from its reliance on the relative position of antenna pairs.
To show how the covariance $c_{n, n^\prime}$ depends on the positional information, we approximate \eqref{eq_covariance_matrix} as
\begin{IEEEeqnarray*}{rcl}
c_{n, n^\prime} & \overset{\text{(a)}} {\approx} & \frac{1}{M} \left[P_\phi (\phi_1), \cdots, P_\phi (\phi_M)\right]
\begin{bmatrix}
\tilde{c}_{n, n^\prime}(\phi_1)\\
\vdots\\
\tilde{c}_{n, n^\prime}(\phi_M)\\
\end{bmatrix}, \IEEEeqnarraynumspace \IEEEyesnumber \label{eq_linear_covmat}
\end{IEEEeqnarray*}
where $P_\phi (\cdot)$ is the PDF of the AoDs, (a) is obtained by approximating the expectation by taking average on samples $\phi_m \triangleq \frac{2 \pi (m - 1)}{M}$, and $M$ is the number of samples of $\phi$ affecting the accuracy of the approximate.

In \eqref{eq_linear_covmat}, the coefficients in $[P_\mathtt{\phi}(\phi_1), \cdots, P_\mathtt{\phi}(\phi_M)]$ reflect the unknown distribution of $\phi$, while the positional information is reflected in $[\tilde{c}_{n, n^\prime}(\phi_1), \cdots, \tilde{c}_{n, n^\prime}(\phi_M)]^\mathsf{T}$.
Hence, we design the embedding function as
\begin{equation}\label{eq_embedding_function}
\mathrm{EB}(n, n^\prime) = [\tilde{c}_{n, n^\prime}(\phi_1), \cdots, \tilde{c}_{n, n^\prime}(\phi_M)].
\end{equation}
Consequently, the resulting $\mathbf{A}^{(l)}$ and the channel covariance matrix $\mathbf{C}$ exhibit the same structure as $\widetilde{\mathbf{C}}(\phi)$.
To see this, we assume $\tilde{c}_{n_1, n_1^\prime}(\phi) = \tilde{c}_{n_2, n_2^\prime}(\phi)$ in $\widetilde{\mathbf{C}}(\phi)$.
Then, we have $c_{n_1, n_1^\prime} = \mathbb{E}[\tilde{c}_{n_1, n_1^\prime}(\phi)] = \mathbb{E}[\tilde{c}_{n_2, n_2^\prime}(\phi)] = c_{n_2, n_2^\prime}$ in $\mathbf{C}$, and $a^{(l)}_{n_1, n_1^\prime} = f_\mathsf{A}^{(l)}(\mathrm{EB}(n_1, n_1^\prime)) = f_\mathsf{A}^{(l)}(\mathrm{EB}(n_2, n_2^\prime)) = a^{(l)}_{n_2, n_2^\prime}$ in $\mathbf{A}^{(l)}$.

According to the definition of $\tilde{c}_{n, n^\prime}(\phi)$, the specific form of the embedding function can be obtained by substituting the steering vectors into \eqref{eq_covariance_matrix}.
For example, when ULA is deployed at the BS, the positional embedding vector is
\begin{equation}
\mathrm{EB}(n, n^\prime) = \left[e^{j 2\pi (n - n^\prime) \frac{d \sin \phi_1}{\lambda}}, \cdots, e^{j 2\pi (n - n^\prime) \frac{d \sin \phi_M}{\lambda}}\right]^\mathsf{T},
\end{equation}
which is a function of $n - n^\prime$.
Consequently, $\mathbf{W}_\mathrm{G}^{(l)}$ is a Toeplitz matrix.
Similarly, if a UPA is deployed, the corresponding embedding function can be derived from the UPA steering vector provided in \cite{chen1993two}.

\subsubsection{Attention Coefficients for Aggregating User Vertices}

Due to the orthogonality of the SRSs, the channel estimation problem in \eqref{eq_channel_estimation_problem} can be decoupled among users, and thus the update of $\mathbf{d}_{n, k}^{(l)}$ does not require aggregating information from edges $(n, k^\prime)$ connecting to other users.
Therefore, we set $\boldsymbol{\omega}^{(l)}_{k, k^\prime, n}$ as zero for learning channel estimation.
The update equation for estimating the $k$-th user's channel vector is
\begin{IEEEeqnarray}{rcl}\label{eq_update_pgnn_ce}
\mathbf{d}_{n, k}^{(l)} & = & f_\mathsf{C}^{(l)}\Big(\mathring{\mathbf{d}}_{n, k}^{(l - 1)}, \sum_{n^\prime = 1}^{N}f_\mathsf{A}^{(l)}(\mathrm{EB}(n, n^\prime)) \odot \mathring{\mathbf{d}}_{n^\prime, k}^{(l - 1)}\Big).\IEEEeqnarraynumspace
\end{IEEEeqnarray}

For learning the E2E precoding policy, we compute $\boldsymbol{\omega}^{(l)}_{k, k^\prime, n}$ in the same way as \cite{liu2024multidimensional} to leverage channel correlation between users.
The resulting update equation of the PA-GNN, which is PE to users, is
\begin{IEEEeqnarray}{rl}\label{eq_update_p2gnn}
\mathbf{d}_{n, k}^{(l)} = & f_\mathsf{C}^{(l)}\Bigg(\mathbf{d}_{n, k}^{(l - 1)}, \sum_{n^\prime = 1}^{N}f_\mathsf{A}^{(l)}(\mathrm{EB}(n, n^\prime)) \odot \mathring{\mathbf{d}}_{n^\prime, k}^{(l - 1)},\IEEEnonumber\\
& \quad \sum_{k^\prime = 1}^K f_\mathsf{U}^{(l)}(\mathbf{d}_k^{(l)}, \mathbf{d}_{k^\prime}^{(l)}) \odot \mathring{\mathbf{d}}^{(l)}_{n, k^\prime}\Bigg),
\end{IEEEeqnarray}
where $f_\mathsf{U}^{(l)}(\mathbf{d}_k^{(l)}, \mathbf{d}_{k^\prime}^{(l)}) \triangleq \sum_{n = 1}^{N} \mathbf{W}^{(l)}_\mathsf{Q}\mathbf{d}_{n, k}^{(l)} \odot \mathbf{W}^{(l)}_\mathsf{K}\mathbf{d}_{n, k^\prime}^{(l)} / N$.

\section{Simulation Results}

In this section, we evaluate the performance of the proposed PA-GNN for learning channel estimation and E2E precoding.

\subsection{Simulation Setup}

Simulation setups are as follows unless otherwise specified.

According to the 3GPP standard \cite{3GPP38211}, each frame consists of 10 subframes, each of which is with duration of $T_\mathrm{s} = 1$ ms.
The numbers of uplink and downlink subframes are set as $L_\mathrm{UL} = L_\mathrm{DL} = 5$ \cite{wang2024learning}.
The carrier frequency of the system is 3.5 GHz.
The uplink and downlink SNRs are denoted as $\mathrm{SNR}_\mathrm{UL} \triangleq \frac{P_\mathrm{UL}}{\sigma_\mathrm{UL}^2}$ and $\mathrm{SNR}_\mathrm{DL} \triangleq \frac{P_\mathrm{DL}}{\sigma_\mathrm{DL}^2}$, respectively.

The channels of different users are independent, which are generated with the time-varying Saleh-Valenzuela channel model \cite{saleh1987statistical}.
The BS is equipped with a ULA, and the channel vector of the $k$-th user in the $l$-th uplink subframe is
\begin{equation}\label{eq_channel_model}
\mathbf{h}_{\mathrm{UL}, k}^l = \frac{1}{\sqrt{N_\mathrm{cl} N_\mathrm{ray}}}\sum_{a = 1}^{N_\mathrm{cl}}\sum_{b = 1}^{N_\mathrm{ray}}\alpha_k^{a, b}e^{j 2 \pi l T_\mathrm{s} f_\mathrm{d}^{a, b}}\mathbf{a}_\mathrm{ULA}(\phi_k^{a, b}),
\end{equation}
where $N_\mathrm{cl}$ and $N_\mathrm{ray}$ are the numbers of clusters and scattering paths in each cluster, respectively, $\alpha_k^{a, b}$ and $\phi_k^{a, b}$ are the complex channel gain and the AoD, respectively, $f_\mathrm{d}^{a, b} \triangleq \frac{v_k}{\lambda} \cos(\theta_\mathrm{d}^{a, b})$ is the Doppler frequency shift, $v_k$ is the speed of the $k$-th user, and $\theta_\mathrm{d}^{a, b}$ is the movement direction relative to the angle of arrival of the $k$-th user.
The settings of these parameters are given in Table \ref{tab_sim_param}.

\begin{table}[htbp]
\caption{Simulation Parameters}
\vspace{-5mm}
\label{tab_sim_param}
\begin{center}
\renewcommand{\arraystretch}{1.2}
\begin{tabular}{|c|c|}
\hline
\textbf{Parameters} & \textbf{Values} \\
\hline
$N_\mathrm{cl}, N_\mathrm{ray}$ & 2, 5 \\
\hline
$\alpha_k^{a, b}$ & $\mathcal{CN}(0, 1)$\\
\hline
$\phi_k^{a, b}$ & $\mathbb{U}(\bar{\phi}_k^a - \Delta_\phi, \bar{\phi}_k^a + \Delta_\phi)$\\
\hline
$\bar{\phi}_k^a$ & $\mathbb{U}(-\pi / 2, \pi / 2)$\\
\hline
$\Delta_\phi$ & $\pi / 36$\\
\hline
$v_k$ & 30 km/h \\
\hline
$\theta_\mathrm{d}^{a, b}$ & $\mathbb{U}(-\pi, \pi)$\\
\hline
\end{tabular}
\end{center}
\end{table}

According to the frame structure shown in Fig.~\ref{fig_frame_structure} and the channel reciprocity, the channel vector of the $k$-th user in the $l$-th downlink subframe is obtained by replacing the index $l$ with $l + L_\mathrm{SRS}$ in \eqref{eq_channel_model}, i.e., $\mathbf{\mathbf{h}}_{\mathrm{DL}, k}^l = \mathbf{\mathbf{h}}_{\mathrm{UL}, k}^{l + L_\mathrm{SRS}}$.

We generate a training set consisting of 100,000 samples, and the fine-tuned hyper-parameters of the proposed PA-GNN are listed in Table \ref{tab_hyperparameter}.
For channel estimation, each sample is a pair of $\mathbf{R}$ and $\mathbf{H}_\mathrm{UL}$. The DNNs  with input $\mathbf{R}$ are trained to minimize the MSE between its output (i.e., the channel estimate) and the label $\mathbf{H}_\mathrm{UL}$ via \textit{supervised learning}.
For E2E precoding, each sample is a realization of $\mathbf{R}$.
The DNNs with input $\mathbf{R}$ are trained to maximize the downlink sum rate via \textit{unsupervised learning}, where the loss function is the negative sum rate.

\begin{table}[htbp]
\caption{Hyper-parameters}
\vspace{-5mm}
\label{tab_hyperparameter}
\begin{center}
\begin{tabular}{|c|c|c|}
\hline
\textbf{Hyper-} & \textbf{Channel} & \textbf{E2E} \\
\textbf{parameters} & \textbf{Estimation} & \textbf{Precoding}\\
\hline
$N_\mathrm{L}$ & 4 & 5 \\
\hline
$C_l$ & \multicolumn{2}{c|}{128} \\
\hline
Hidden Layer Activation & \multicolumn{2}{c|}{Leaky ReLU}\\
\hline
Output Layer Activation & Linear & $\frac{\sqrt{P_\mathrm{DL}} \mathbf{D}^{(N_\mathrm{L})}}{||\mathbf{D}^{(N_\mathrm{L})}||_\mathrm{F}}$\\
\hline
$f_\mathsf{P}^{(l)}(\cdot)$ & - & Linear \\
\hline
$f_\mathsf{C}^{(l)}(\cdot)$ & \multicolumn{2}{c|}{Single-layer FNN} \\
\hline
$f_\mathsf{A}^{(l)}(\cdot)$ & \multicolumn{2}{c|}{Single-layer FNN ($M = 128$)}\\
\hline
Batch Size & \multicolumn{2}{c|}{256}\\
\hline
Learning Rate & \multicolumn{2}{c|}{0.01}\\
\hline
\multicolumn{3}{l}{* The input dimension of $f_\mathsf{C}^{(l)}(\cdot)$ and the output dimensions}\\
\multicolumn{3}{l}{\textcolor{white}{*} of $f_\mathsf{C}^{(l)}(\cdot)$ and $f_\mathsf{A}^{(l)}(\cdot)$ are determined by $C_l$.}
\end{tabular}
\end{center}
\end{table}
\vspace{-5mm}

\subsection{Channel Estimation}\label{subsection_sim_ce}

For channel estimation, we consider $K=1$ since the problem can be decoupled among different users, and set $L_\mathrm{SRS} = Q = 1$ to estimate the channel in a single subframe.
We compare the PA-GNN with the following baselines.
\begin{itemize}
\item \textbf{LS:} The channel estimator that regards $\mathbf{R}$ in \eqref{eq_avg_effective_srs} as the estimate of $\mathbf{H}_\mathrm{UL}$.

\item \textbf{LMMSE:} The LMMSE estimator in \eqref{eq_lmmse}, where $\mathbf{C}$ is obtained from \eqref{eq_linear_covmat} with $M = 512$.

\item \textbf{OMP:} The channel estimation algorithm used in \cite{liu2024sparse}, where the dictionary contains $N$ steering vectors uniformly distributed in the angular domain.

\item \textbf{Vanilla GNN:} The GNN designed in \cite{zhao2024understanding} without attention, which is PE to both users and antennas.\footnote{The PE property of users is unnecessary when $K=1$.}

\item \textbf{CNN with Attention:} The DNN proposed in \cite{gao2021attention}, which integrates CNN with an attention mechanism.

\item \textbf{Dual CNN:} The DNN proposed in \cite{jiang2021dual}, where two CNNs learning in spatial and angular domains, respectively, are cascaded together for channel estimation.

\item \textbf{Transformer:} The encoder part of the Transformer architecture in \cite{vaswani2017attention}.
The effective SRSs received by different antennas are regarded as different tokens.
Positional encoding is reserved to capture the spatial correlation.

\item \textbf{Transformer (w/o p.e.)} The only different to Transformer is removing the positional encoding, and it is PE to antennas.
\end{itemize}

The baseline DNNs, except for the Dual CNN, have the same numbers of hidden layers and the same hidden dimensions as the PA-GNN for a fair comparison.
For the Dual CNN, each of the two CNNs consists of four hidden layers, with each layer containing 32 neurons, as in \cite{gao2021attention}.

\subsubsection{Performance versus System and Environment Parameters}

For a fair comparison, we temporarily do not evaluate the performance of the Dual CNN, since it leverages information from angular domain while other DNNs only learn in spatial domain.
The performance of DNNs when learning in angular domain will be compared later.

In Fig~\ref{fig_ce_snr}, we compare the channel estimation accuracy of different methods under different uplink SNRs.
It is shown that the performance of vanilla GNN and Transformer (w/o p.e.) nearly overlaps with that of LMMSE, whereas other learning-based channel estimators outperform non-learning methods.
The performance of vanilla GNN and Transformer (w/o p.e.) is inferior since they can only learn PE policies with respect to antennas, whereas the original channel estimation policy is not PE. 
According to our previous analysis, positional information should be incorporated so that the network can learn the corresponding augmented policy.
With positional encoding, the Transformer can leverage antenna positional information and thus achieves improved performance.
The proposed PA-GNN is with the lowest MSE, since its positional attention mechanism incorporates steering vectors as prior knowledge.

\begin{figure}[htbp]
\centering
\includegraphics[width=0.9\linewidth]{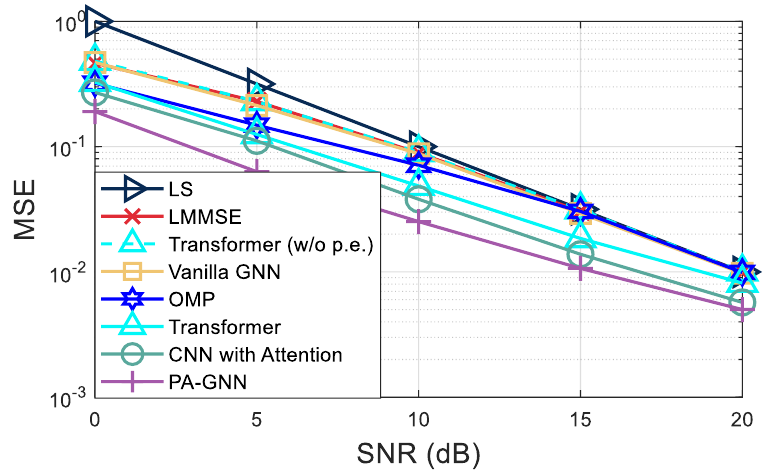}\vspace{-2mm}
\caption{MSE versus $\text{SNR}_\mathrm{UL}$, $N = 128$, $N_\mathrm{cl} = 2$.}
\label{fig_ce_snr}
\end{figure}

In Fig.~\ref{fig_ce_ant}, we show the impact of the number of antennas. It is shown that the MSEs of all methods decrease with more antennas, since more information can be used for estimating the channel coefficient of each antenna due to the spatial correlation.
Again, the proposed PA-GNN achieves the lowest MSE among all methods, and the performance gain over other DNNs comes from the proper inductive bias.

\begin{figure}[htbp]
\centering
\includegraphics[width=0.9\linewidth]{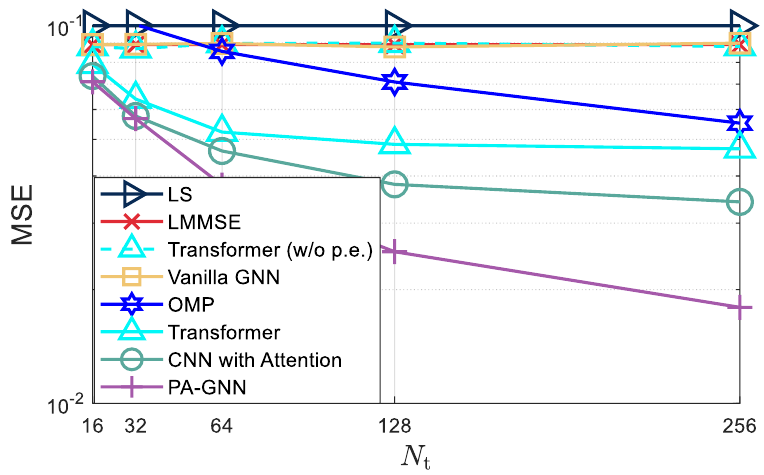}\vspace{-2mm}
\caption{MSE versus $N$, $\mathrm{SNR}_\mathrm{UL} = 10$ dB, $N_\mathrm{cl} = 2$.}
\label{fig_ce_ant}
\includegraphics[width=0.9\linewidth]{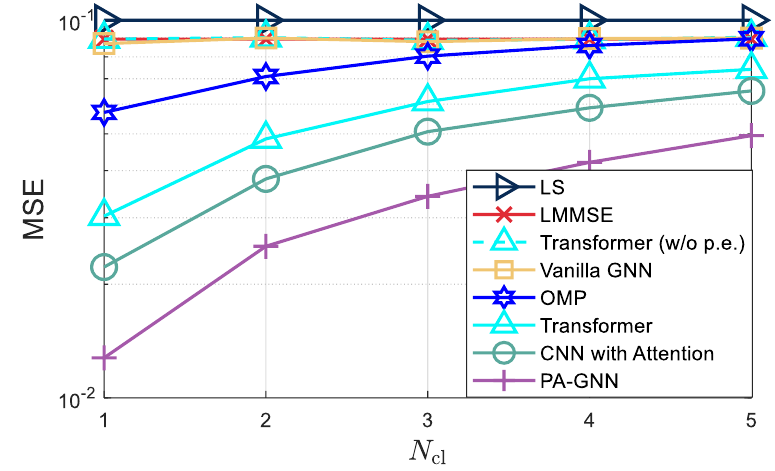}\vspace{-2mm}
\caption{MSE versus $N_\mathrm{cl}$, $\mathrm{SNR}_\mathrm{UL} = 10$ dB, $N = 128$.}
\label{fig_ce_cluster}
\end{figure}

In Fig.~\ref{fig_ce_cluster}, we show the impact of the channel sparsity.
We can see that the MSEs of all methods are lower when the number of clusters is small.
This is because the spatial correlation is stronger in a sparser channel, which leads to a more accurate estimate. The PA-GNN achieves significant lower MSE than other methods.

\subsubsection{Performance of Learning in Angular Domain}

In previous simulations, the inputs of the DNNs are the effective received SRS $\mathbf{R}$.
As suggested in \cite{gao2021attention, jiang2021dual}, we can also learn the channel estimation policy in angular domain by inputting the discrete Fourier transform of $\mathbf{R}$ into the DNNs, and the channel estimate is then obtained by applying the inverse Fourier transform to their outputs.

The estimation accuracies in spatial and angular domains are compared in Fig.~\ref{fig_ce_domain}, and the results of the vanilla GNN and Transformer (w/o p.e.) are not provided due to their inferior performance.\footnote{``(spatial)'' and ``(angle)'' in the legends of Fig.~\ref{fig_ce_domain} refer to learning in spatial and angular domains, respectively.}
We can see that compared to learning in spatial domain, a higher accuracy can be achieved when learning in angular domain, and the gain is more obvious with larger $N$.
When learning in angular domain, the proposed PA-GNN still achieves the best performance among all DNNs, and the Transformer can achieve nearly identical performance.
To further demonstrate the advantage of the proposed PA-GNN, we consider a scenario where the BS is equipped with a UPA, and the channel estimation error achieved by the PA-GNN is reduced by 16\% compare to the Transformer when $N = 128$.
The performance degradation of the Transformer stem from its positional encoding being designed for sequential data, which aligns with the structure of ULAs but not with UPAs.

\begin{figure}[htbp]
\centering
\includegraphics[width=0.9\linewidth]{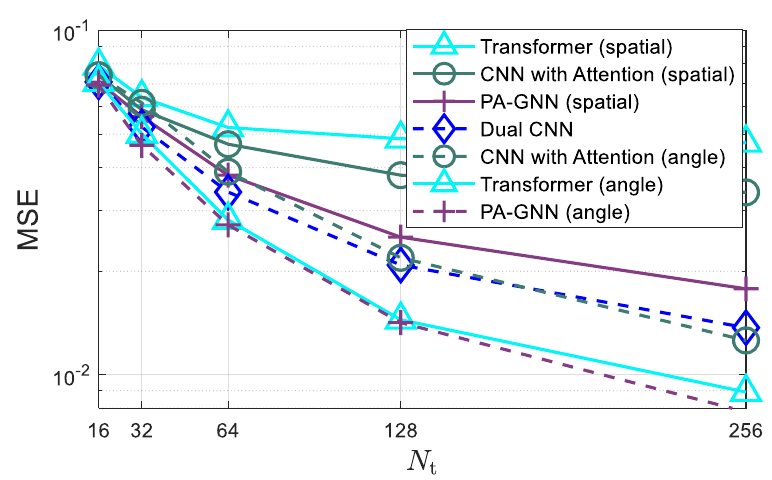}\vspace{-2mm}
\caption{MSE versus $N$ when learning in different domains, $\mathrm{SNR}_\mathrm{UL} = 10$ dB, $N_\mathrm{cl} = 2$.}
\label{fig_ce_domain}
\end{figure}

\subsection{E2E Precoding}

For E2E precoding, we take the maximization of the sum rate of the second downlink subframe as an example.
Unless otherwise specified, we set $L_\mathrm{SRS} = Q = 2$, and $\text{SNR}_\mathrm{UL} = \text{SNR}_\mathrm{DL} = 10$ dB. We compare the PA-GNN with the following baselines.
\begin{itemize}
\item \textbf{WMMSE (perf.):} The precoding matrix is obtained by the WMMSE algorithm in \cite{shi2011iteratively} with perfect channel matrix.
It serves as an \textit{unachievable} performance upper bound for comparison.

\item \textbf{WMMSE (outdate):} The precoding matrix is obtained by the WMMSE algorithm \cite{shi2011iteratively} with the LMMSE estimate of the channel matrix $\mathbf{H}_\mathrm{UL}^{L_\mathrm{SRS}}$.

\item \textbf{GNN-2D:} The GNN proposed in \cite{liu2024multidimensional}. The only difference from the PA-GNN is to use identical weights for aggregating antenna vertices, whose update equation can be obtained by removing $f^{(l)}_\mathsf{A}(\mathrm{EB}(n, n^\prime))$ in \eqref{eq_update_p2gnn}.
The policy learned by the GNN-2D is PE to both antennas and users.

\item \textbf{GNN-1D2D:} The GNN proposed in \cite{zhao2023learning}, which cascades two GNNs for learning the E2E precoding policy.
The former GNN is PE to users without introducing parameter sharing among antenna vertices, while the later GNN is the GNN-2D.
The policy learned by the GNN-1D2D is only PE to users.

\item \textbf{Transformer:} The same Transformer as used in Section \ref{subsection_sim_ce}, with an extra normalization step to enforce the power constraint on the outputs.
We regard the effective SRSs received by different antennas as different tokens and employ positional encoding along the antenna dimension for comparison.

\end{itemize}

\subsubsection{Performance versus System and Environment Parameters}

In Fig.~\ref{fig_e2e_snr}, we compare the sum rate achieved by different methods under different downlink SNRs.
Only partial results of WMMSE (perf.) are provided, in order to show the performance difference among other methods when the downlink SNR is high.
It can be observed that all learning-based methods outperform WMMSE (outdated).
This is because the DNNs for E2E learning can predict channels implicitly by leveraging temporal correlation, and achieve lower MSE than the LMMSE channel estimator.
Meanwhile, the performance of the WMMSE algorithm is very sensitive to the outdating of channels.
Moreover, the proposed PA-GNN performs the best, and the gains over other DNNs grow as the downlink SNR increases.

\begin{figure}[htbp]
\centering
\includegraphics[width=0.9\linewidth]{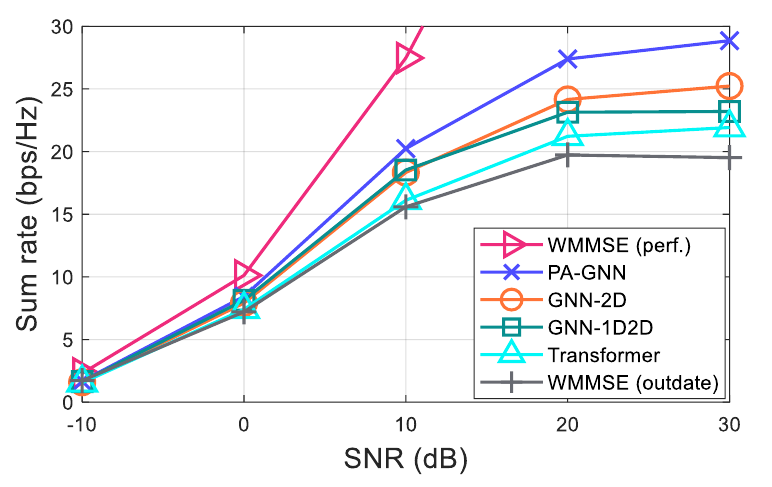}\vspace{-2mm}
\caption{Sum rate versus $\text{SNR}_\mathrm{DL}$, $N = 16$, $K = 8$.}
\label{fig_e2e_snr}
\end{figure}

\begin{figure}[htbp]
\centering
\includegraphics[width=0.9\linewidth]{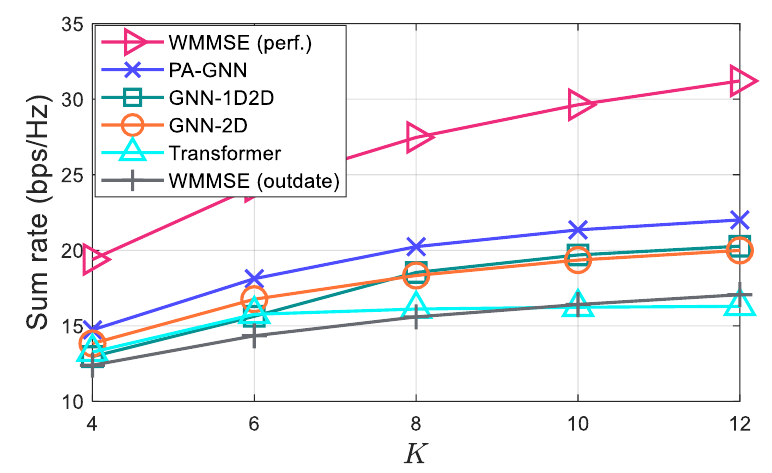}\vspace{-2mm}
\caption{Sum rate versus $K$, $N = 16$.}
\label{fig_e2e_user}
\includegraphics[width=0.9\linewidth]{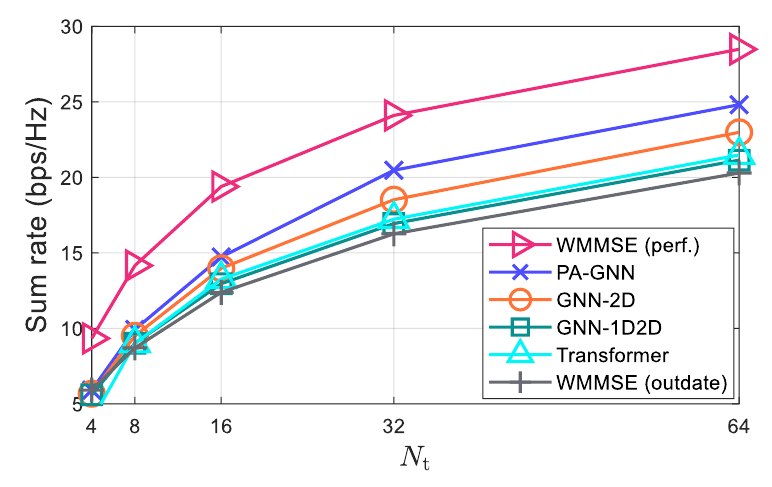}\vspace{-2mm}
\caption{Sum rate versus $N$, $K = 4$.}
\label{fig_e2e_ant}
\end{figure}

In Fig.~\ref{fig_e2e_user} and Fig.~\ref{fig_e2e_ant}, we provide the sum rates versus the numbers of users and antennas.
We can see that the proposed PA-GNN achieves the best performance among all the methods except the unachievable upper bound.
The GNN-1D2D and GNN-2D are inferior to the PA-GNN for different reasons.
Specifically, GNN-1D2D does not share parameters among antenna vertices in its initial layers and thus has larger hypothesis space, leading to an overfitting for channel estimation \cite{zhao2024understanding}.
GNN-2D, by contrast, enforces permutation equivariance along the antenna dimension without incorporating positional information, and therefore cannot represent the statistical distinctions among antennas required by the E2E policy.
The performance of the Transformer is poor when $K$ and $N$ are close.
This is because, under such scenarios, the precoding policy is hard to learn due to severe multi-user interference, while the Transformer applying attention along the antenna dimension fails to capture the correlation among users.

\begin{figure}[htbp]
\centering
\includegraphics[width=0.9\linewidth]{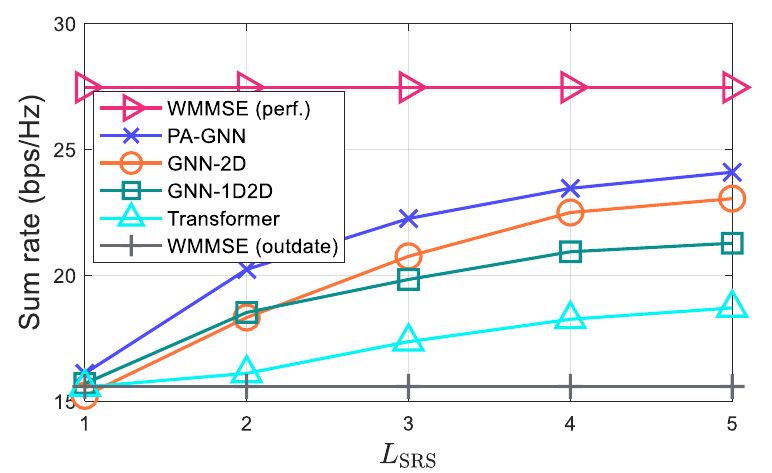}\vspace{-2mm}
\caption{Sum rate versus $L_\mathrm{SRS}$, $N = 16$, $K = 8$.}
\label{fig_e2e_srs}
\end{figure}

In Fig.~\ref{fig_e2e_srs}, we evaluate the sum rates versus the number of uplink subframes for transmitting SRSs.
It can be observed that the performance of all learning-based methods increases with $L_\mathrm{SRS}$, since a more accurate channel matrix can be obtained with more historical information.
Nonetheless, there is still a gap to the performance upper bound even if larger $L_\mathrm{SRS}$ is considered (e.g., the performance of the PA-GNN is no more than 25 bps/Hz), because the future channels can not be predicted perfectly.

\subsubsection{Sample Complexity}

In Fig.~\ref{fig_sample_complexity}, we compare the sample complexities of the GNNs by evaluating their performance with different numbers of training samples.

As can be seen, when the number of training samples is large (say more than 50,000), the sum rate achieved by the Transformer is the lowest, since it can not mitigate the multi-user interference effectively as mentioned before.
The GNN-2D also performs poorly, since it can only learn 2D-PE functions, and thus its hypothesis space does not include the original E2E precoding policy without positional augmentation.

When the training samples are fewer, the performance of GNN-1D2D and Transformer degrades rapidly.
This is because GNN-1D2D or Transformer does not introduce parameter sharing among antenna or user vertices, respectively, and thus has larger hypothesis space.
As a result, more training samples are required to achieve a satisfactory performance \cite{zhao2024understanding}.

\begin{figure}[htbp]
\centering
\includegraphics[width=0.9\linewidth]{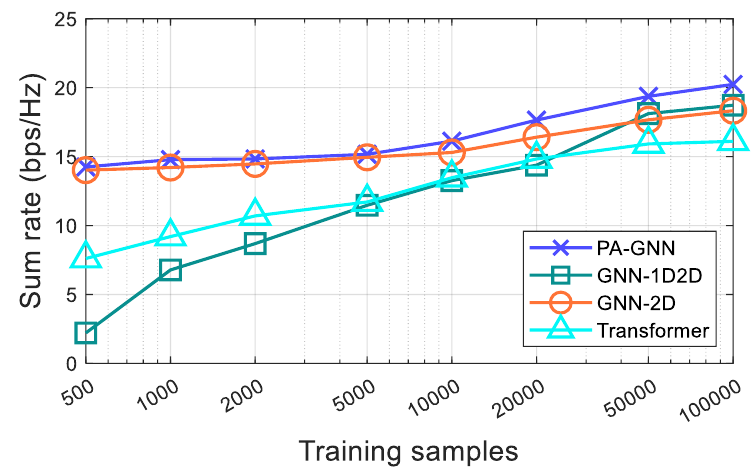}\vspace{-2mm}
\caption{Sample complexity of the GNNs, $N = 16, K = 8$.}
\label{fig_sample_complexity}
\end{figure}

\subsubsection{Size Generalizability}

We evaluate the generalizability of the GNNs to the numbers of transmit antennas and users, by testing the GNNs in the scenarios ``unseen'' by the training samples.
The performance is measured by the sum rate achieved on test samples from the new environment, normalized by the performance of a PA-GNN trained and evaluated within the new environment (with legend ``Ref. Perf.'').

In Table \ref{tab_dim_gen_user}, we compare the generalizability of the GNNs to the number of users, where the GNNs are trained with the samples in a scenario of $K = 8$.
The Transformer is not evaluated, since we treat the effective SRSs of different antennas as distinct input tokens, and thus it is not generalizable to $K$.
We can see that the PA-GNN achieves better generalization performance than GNN-1D2D and GNN-2D.

\begin{table}[htbp]
\caption{Generalization ability to $K$, $N = 16$.}
\vspace{-5mm}
\label{tab_dim_gen_user}
\begin{center}
\begin{tabular}{|c|c|c|c|c|c|}
\hline
$K$ & 4 & 6 & 8 (Train) & 10 & 12\\
\hline
\textbf{Ref. Perf.} & \multirow{2}{*}{14.7} & \multirow{2}{*}{18.1} & \multirow{2}{*}{20.2} & \multirow{2}{*}{21.4} & \multirow{2}{*}{22.0}\\
\textbf{(bps/Hz)} & & & & & \\
\hline
\textbf{PA-GNN} & \textbf{90.3\%} & \textbf{100.0\%} & \textbf{100.0\%} & \textbf{99.4\%} & \textbf{98.1\%}\\
\hline
\textbf{GNN-1D2D} & 86.1\% & 92.8\% & 91.5\% & 89.9\% & 88.1\%\\
\hline
\textbf{GNN-2D} & 53.0\% & 90.6\% & 90.6\% & 90.0\% & 89.0\%\\
\hline
\end{tabular}
\end{center}
\end{table}\vspace{-2mm}

In Table \ref{tab_dim_gen_antenna}, we compare the generalizability of the GNNs to the number of antennas, where the GNNs are trained with the samples in a scenario of $N = 16$.
The GNN-1D2D is not evaluated, since it is not generalizable to $N$.
It is shown that both of the PA-GNN and GNN-2D can be well generalized except the case of $N = K = 4$, and the PA-GNN performs slightly better.
The generalization performance of the Transformer is worse.

\begin{table}[htbp]
\caption{Generalization ability to $N$, $K = 4$.}
\vspace{-5mm}
\label{tab_dim_gen_antenna}
\begin{center}
\begin{tabular}{|c|c|c|c|c|c|}
\hline
$N$ & 4 & 8 & 16 (Train) & 32 & 64\\
\hline
\textbf{Ref. Perf.} & \multirow{2}{*}{5.9} & \multirow{2}{*}{9.9} & \multirow{2}{*}{14.7} & \multirow{2}{*}{20.5} & \multirow{2}{*}{24.8}\\
\textbf{(bps/Hz)} & & & & & \\
\hline
\textbf{PA-GNN} & \textbf{86.1\%} & \textbf{91.9\%} & \textbf{100.0\%} & \textbf{93.2\%} & \textbf{92.7\%}\\
\hline
\textbf{GNN-2D} & 85.8\% & 91.1\% & 95.1\% & 90.0\% & 90.7\%\\
\hline
\textbf{Transformer} & 75.0\% & 87.2\% & 91.2\% & 87.0\% & 88.7\% \\
\hline
\end{tabular}
\end{center}
\end{table}

\subsubsection{Learned Attention Matrix}

To understand the behavior of the PA-GNN for learning E2E precoding policy, we virtualize the attention matrix $\mathbf{A}^{(l)}$ learned by the PA-GNN in the form of heat map, as shown in Fig.~\ref{fig_attention_matrix}.

\begin{figure}[htbp]
\centering
\includegraphics[width=0.9\linewidth]{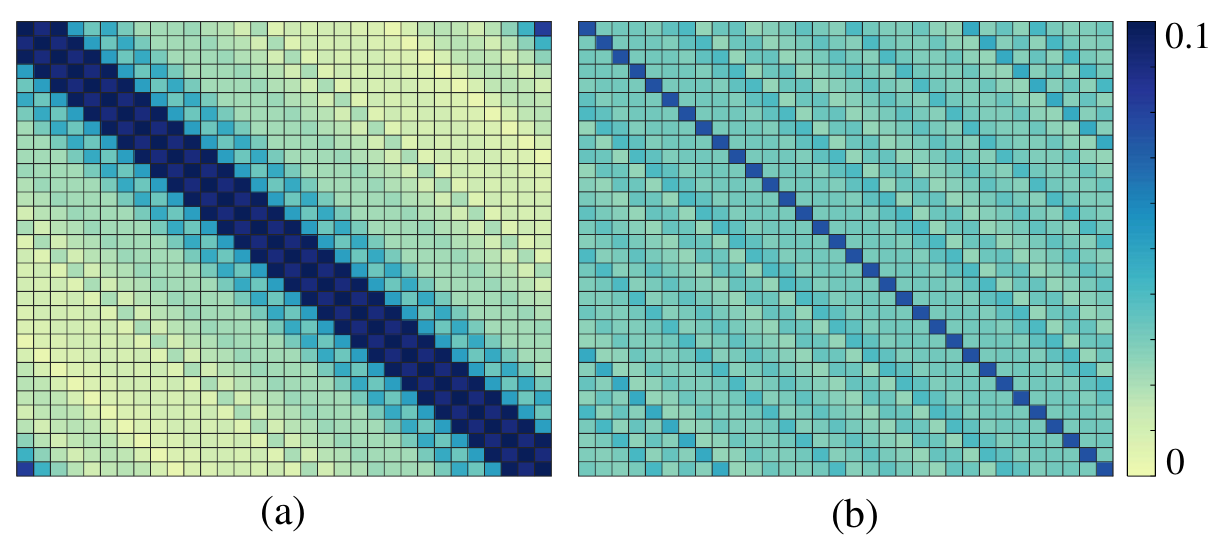}\vspace{-2mm}
\caption{Attention matrices learned by the PA-GNN in the first layer (a) and the last layer (b).}
\label{fig_attention_matrix}
\end{figure}

As can be seen, in the first layer of the PA-GNN, the elements near the diagonal are much larger, exhibiting a structure similar to that shown in Fig.~\ref{fig_cor_mat_structure}(b). This indicates that the network mainly aggregates information from nearby antennas. Such a behavior is consistent with the spatial correlation induced by the antenna array geometry, where antennas that are physically close tend to experience similar propagation environments. Consequently, the positional embedding plays an important role in guiding the aggregation process in the early layers.

In contrast, in the last layer, the off-diagonal elements of the learned attention matrix take similar values, resulting in a structure close to that shown in Fig.~\ref{fig_cor_mat_structure}(a). This suggests that the dependence on positional embeddings weakens in deeper layers. This can be explained by the structure of the E2E precoding policy, which can be viewed as a composite of the channel estimation/prediction policy and the precoding policy. The latter, as discussed in \textbf{P2}, is distribution-independent. As a result, the subsequent layers mainly learn the PE precoding policy and therefore do not require positional information.

\section{Conclusion}

In this paper, we proposed a PA-GNN for learning distribution-dependent policies in MU-MISO systems.
We first analyzed how the mathematical expectation in the objective of wireless problems lead to policies whose permutation properties depend on the underlying distribution.
Then, we proposed a novel attention mechanism where the attention coefficients are determined by the relative positions of vertices, allowing the GNN to represent permutation non-equivariant functions.
We applied the PA-GNN to two tasks, i.e.,  channel estimation and E2E precoding, whose policies are not PE to antennas under spatially correlated channels, while the same idea can be extended to other dimensions, such as the subcarrier dimension in multi-carrier systems.
Simulation results demonstrated the superiority of the PA-GNN over existing numerical and learning-based methods in terms of performance, and it requires fewer training samples than other learning-based counterparts.
Furthermore, the PA-GNN can be generalized to systems with varying numbers of users and antennas.

\appendices

\section{Proof of Proposition \ref{proposition_pe}}\label{app_proposition_pe}
We take the permutation to the antennas as an example for the proof.
Since it has been proved in \cite{eisen2020optimal} that a policy is PE if the objectives and the constraints are PI, we only need to prove that $U(\mathbf{X}, \mathbf{Y}) = \mathbb{E}_{\mathtt{Z} | \mathtt{X}}[U^\prime(\mathbf{Y}, \mathbf{Z}) | \mathbf{X}]$ is PI to antennas.

When $\mathbf{X}$ and $\mathbf{Y}$ are permuted along the antenna dimension by $\boldsymbol{\Pi}_\mathrm{A}$, we have
\begin{IEEEeqnarray*}{rl}
& U(\boldsymbol{\Pi}_\mathrm{A}^\mathsf{T}\mathbf{X}, \boldsymbol{\Pi}_\mathrm{A}^\mathsf{T}\mathbf{Y})\\
= & \mathbb{E}_{\mathtt{Z} | \mathtt{X}}\left[U^\prime(\boldsymbol{\Pi}_\mathrm{A}^\mathsf{T}\mathbf{Y}, \mathbf{Z}) | \boldsymbol{\Pi}_\mathrm{A}^\mathsf{T}\mathbf{X}\right]\\
\overset{\text{(a)}}{=} & \int U^\prime(\boldsymbol{\Pi}_\mathrm{A}^\mathsf{T}\mathbf{Y}, \mathbf{Z}) P_{\mathtt{Z} | \mathtt{X}}(\mathbf{Z} | \boldsymbol{\Pi}_\mathrm{A}^\mathsf{T}\mathbf{X}) \ d\mathbf{Z}\\
\overset{\text{(b)}}{=} & \int U^\prime(\boldsymbol{\Pi}_\mathrm{A}^\mathsf{T}\mathbf{Y}, \boldsymbol{\Pi}_\mathrm{A}^\mathsf{T}\mathbf{Z}^\prime) P_{\mathtt{Z} | \mathtt{X}}(\boldsymbol{\Pi}_\mathrm{A}^\mathsf{T}\mathbf{Z}^\prime | \boldsymbol{\Pi}_\mathrm{A}^\mathsf{T}\mathbf{X}) \ d\boldsymbol{\Pi}_\mathrm{A}^\mathsf{T}\mathbf{Z}^\prime\\
\overset{\text{(c)}}{=} & \int U^\prime(\mathbf{Y}, \mathbf{Z}^\prime) P_{\mathtt{Z} | \mathtt{X}}(\mathbf{Z}^\prime | \mathbf{X}) \ d\mathbf{Z}^\prime = U(\mathbf{X}, \mathbf{Y}),
\end{IEEEeqnarray*}
where (b) is obtained by substituting $\mathbf{Z}^\prime \triangleq \boldsymbol{\Pi}_\mathrm{A} \mathbf{Z}$, and (c) holds because $U^\prime(\mathbf{Y}, \mathbf{Z})$ and $P_{\mathtt{Z} | \mathtt{X}} (\mathbf{Z} | \mathbf{X})$ are PI to antennas and the permutation matrices  are orthogonal matrices.

In the derivation above, the permutations $\boldsymbol{\Pi}_\mathrm{A}^\mathsf{T} \mathbf{X}$ and $\boldsymbol{\Pi}_\mathrm{A}^\mathsf{T} \mathbf{Y}$ represent that the measurements and decisions associated with different antennas are permuted, rather than a mere reordering of antenna indices.
Consequently, the conditional PDF $P_{\mathtt{Z} | \mathtt{X}} (\cdot | \cdot)$ used in step (a) is the same as that defined in \eqref{eq_general_optimization}.
Otherwise, if the permutation is caused by index reordering, the conditional PDF will transform into a permuted version $P_{\mathtt{Z}^\pi | \mathtt{X}^\pi} (\mathbf{z} | \mathbf{x}) \triangleq P_{\mathtt{Z} | \mathtt{X}} (\boldsymbol{\Pi}_\mathrm{A}^\mathsf{T}\mathbf{z} | \boldsymbol{\Pi}_\mathrm{A}^\mathsf{T}\mathbf{x})$ accordingly as in \eqref{eq_pdf_permutation_relation}.
In that case, the invariance of the utility always holds and does not rely on the PI properties of $U^\prime(\mathbf{Y}, \mathbf{Z})$ and $P_{\mathtt{Z} | \mathtt{X}} (\mathbf{Z} | \mathbf{X})$, as the physical scenario is identical after reordering antenna indices.

The proof for the permutation of users is similar.

\section{Proof of Proposition \ref{prop_pe_ant}}\label{proof_prop_pe_ant}

If $P_{\mathtt{h}_\mathrm{DL} | \mathtt{h}_\mathrm{UL}}(\cdot | \cdot)$ is PI to antennas, according to Bayesian theorem, we have
\begin{IEEEeqnarray*}{rl}
& {P}_{\mathtt{h}_\mathrm{UL} | \mathtt{r}}(\boldsymbol{\Pi}_\mathrm{A}\mathbf{h}_\mathrm{UL} | \boldsymbol{\Pi}_\mathrm{A}\mathbf{r})\\
= & \frac{{P}_{\mathtt{r} | \mathtt{h}_\mathrm{UL}}(\boldsymbol{\Pi}_\mathrm{A}\mathbf{r} | \boldsymbol{\Pi}_\mathrm{A}\mathbf{h}_\mathrm{UL}){P}_{\mathtt{h}_\mathrm{UL}}(\boldsymbol{\Pi}_\mathrm{A}\mathbf{h}_\mathrm{UL})}{{P}_{\mathtt{r}}(\boldsymbol{\Pi}_\mathrm{A}\mathbf{r})}\\
\overset{\text{(a)}}{=} & \frac{{P}_{\mathtt{n}}\left(\boldsymbol{\Pi}_\mathrm{A}(\mathbf{r} - \mathbf{h}_\mathrm{UL})\right){P}_{\mathtt{h}_\mathrm{UL}}(\boldsymbol{\Pi}_\mathrm{A}\mathbf{h}_\mathrm{UL})}{\int{P}_{\mathtt{n}}\left(\boldsymbol{\Pi}_\mathrm{A}(\mathbf{r} - \mathbf{h}_\mathrm{UL})\right){P}_{\mathtt{h}_\mathrm{UL}}(\boldsymbol{\Pi}_\mathrm{A}\mathbf{h}_\mathrm{UL})\ d\mathbf{h}_\mathrm{UL}}\\
\overset{\text{(b)}}{=} & \frac{{P}_{\mathtt{n}}\left(\mathbf{r} - \mathbf{h}_\mathrm{UL}\right){P}_{\mathtt{h}_\mathrm{UL}}(\mathbf{h}_\mathrm{UL})}{\int{P}_{\mathtt{n}}\left(\mathbf{r} - \mathbf{h}_\mathrm{UL}\right){P}_{\mathtt{h}_\mathrm{UL}}(\mathbf{h}_\mathrm{UL})\ d\mathbf{h}_\mathrm{UL}}\\
= & \frac{{P}_{\mathtt{r} | \mathtt{h}_\mathrm{UL}}(\mathbf{r} | \mathbf{h}_\mathrm{UL}){P}_{\mathtt{h}_\mathrm{UL}}(\mathbf{h}_\mathrm{UL})}{{P}_{\mathtt{r}}(\mathbf{r})} = {P}_{\mathtt{h}_\mathrm{UL} | \mathtt{r}}(\mathbf{h}_\mathrm{UL} | \mathbf{r}),\IEEEyesnumber\label{eq_pe_cpdf}
\end{IEEEeqnarray*}
where (a) holds due to \eqref{eq_avg_effective_srs}, and (b) holds due to the assumption of ${P}_{\mathtt{h}_\mathrm{UL}}(\mathbf{h}_\mathrm{UL}) = {P}_{\mathtt{h}_\mathrm{UL}}(\boldsymbol{\Pi}_\mathrm{A}\mathbf{h}_\mathrm{UL})$.

If $P_{\mathtt{h}_\mathrm{DL} | \mathtt{h}_\mathrm{UL}}(\cdot | \cdot)$ is also PI to antennas, we have
\begin{IEEEeqnarray*}{rl}
& P_{\mathtt{h}_\mathrm{DL} | \mathtt{r}}(\boldsymbol{\Pi}_\mathrm{A}^\mathsf{T}\mathbf{h}_\mathrm{DL} | \boldsymbol{\Pi}_\mathrm{A}^\mathsf{T}\mathbf{r})\\
= & \int P_{\mathtt{h}_\mathrm{DL} | \mathtt{h}_\mathrm{UL}}(\boldsymbol{\Pi}_\mathrm{A}^\mathsf{T}\mathbf{h}_\mathrm{DL} | \mathbf{h}_\mathrm{UL}) P_{\mathtt{h}_\mathrm{UL} | \mathtt{r}}(\mathbf{h}_\mathrm{UL} | \boldsymbol{\Pi}_\mathrm{A}^\mathsf{T}\mathbf{r})\ d\mathbf{h}_\mathrm{UL}\\
\overset{\text{(a)}}{=} & \int P_{\mathtt{h}_\mathrm{DL} | \mathtt{h}_\mathrm{UL}}(\boldsymbol{\Pi}_\mathrm{A}^\mathsf{T}\mathbf{h}_\mathrm{DL} | \boldsymbol{\Pi}_\mathrm{A}^\mathsf{T}\mathbf{h}_\mathrm{UL}^\prime) \\
& \quad \quad \times P_{\mathtt{h}_\mathrm{UL} | \mathtt{r}}(\boldsymbol{\Pi}_\mathrm{A}^\mathsf{T}\mathbf{h}_\mathrm{UL}^\prime | \boldsymbol{\Pi}_\mathrm{A}^\mathsf{T}\mathbf{r})\ d\boldsymbol{\Pi}_\mathrm{A}^\mathsf{T}\mathbf{h}_\mathrm{UL}^\prime\\
\overset{\text{(b)}}{=} & \int P_{\mathtt{h}_\mathrm{DL} | \mathtt{h}_\mathrm{UL}}(\mathbf{h}_\mathrm{DL} | \mathbf{h}_\mathrm{UL}^\prime) P_{\mathtt{h}_\mathrm{UL} | \mathtt{r}}(\mathbf{h}_\mathrm{UL}^\prime | \mathbf{r})\ d\mathbf{h}_\mathrm{UL}^\prime \\
= & P_{\mathtt{h}_\mathrm{DL} | \mathtt{r}}(\mathbf{h}_\mathrm{DL} | \mathbf{r}),
\end{IEEEeqnarray*}
where (a) is obtained by substituting $\mathbf{h}_\mathrm{UL}^\prime \triangleq \boldsymbol{\Pi}_\mathrm{A}^\mathsf{T} \mathbf{h}_\mathrm{UL}$, and (b) holds due to the PI assumption of $P_{\mathtt{h}_\mathrm{DL} | \mathtt{h}_\mathrm{UL}}(\cdot | \cdot)$.

\bibliographystyle{IEEEtran}
\bibliography{reference}

\end{document}